\documentclass[a4paper,twocolumn,11pt,unpublished]{quantumarticle}
\pdfoutput=1
\usepackage[utf8]{inputenc}
\usepackage[english]{babel}
\usepackage[T1]{fontenc}
\usepackage{amsmath}
\usepackage{hyperref}
\usepackage[numbers]{natbib}

\usepackage{tikz}
\usepackage{lipsum}

\newtheorem{theorem}{Theorem}
\usepackage{soul}
\usepackage{graphicx}
\usepackage[caption=false]{subfig}
\usepackage{quantikz}
\usepackage{cleveref}
\usepackage{makecell}
\usepackage{booktabs}  
\usepackage{multirow} 
\usepackage{amsthm}
\usepackage{amssymb}
\usepackage{placeins}

\usepackage{pgfplots}
\pgfplotsset{compat=1.18}

\newtheorem{proposition}[theorem]{Proposition}

\newcolumntype{C}{>{$}c<{$}} 

\definecolor{midred}{HTML}{997700}   
\definecolor{midgold}{HTML}{762A83}  

\definecolor{oiOrange}{HTML}{EE7733}
\definecolor{oiSkyBlue}{HTML}{0077BB}
\definecolor{oiBlueGreen}{HTML}{33BBEE}
\definecolor{oiYellow}{HTML}{EE3377}
\definecolor{oiBlue}{HTML}{CC3311}
\definecolor{oiVermilion}{HTML}{009988}
\definecolor{oiPurple}{HTML}{555555}
\pgfplotscreateplotcyclelist{OkabeIto}{
  {color=oiOrange,    mark=*,       mark options={fill=oiOrange}},
  {color=oiSkyBlue,   mark=square*, mark options={fill=oiSkyBlue}},
  {color=oiBlueGreen, mark=triangle*, mark options={fill=oiBlueGreen}},
  {color=oiYellow,    mark=diamond*, mark options={fill=oiYellow}},
  {color=oiBlue,      mark=pentagon*, mark options={fill=oiBlue}},
  {color=oiVermilion, mark=star,     mark options={fill=oiVermilion}},
  {color=oiPurple,    mark=otimes*,  mark options={fill=oiPurple}},
  {color=black,       mark=+,       mark options={fill=black}},
}

\begin{document}

\title{Efficient Quantum Convolutional Neural Networks for Image Classification: Overcoming Hardware Constraints}

\author{Peter Röseler}
\email{p.roeseler@fz-juelich.de}
\affiliation{Department of Computer Science, University of Bonn, 53111 Bonn, Germany}
\affiliation{Bayer AG, Kaiser-Wilhelm-Allee 1
51373 Leverkusen, Germany}
\affiliation{Jülich Supercomputing Centre, Institute for Advanced Simulation, Forschungszentrum Jülich, 52425 Jülich, Germany}

\author{Oliver Schaudt}
\affiliation{Bayer AG, Kaiser-Wilhelm-Allee 1
51373 Leverkusen, Germany}

\author{Helmut Berg}
\affiliation{Bayer AG, Kaiser-Wilhelm-Allee 1
51373 Leverkusen, Germany}

\author{Christian Bauckhage}
\affiliation{Department of Computer Science, University of Bonn, 53111 Bonn, Germany}

\author{Matthias Koch}
\affiliation{Bayer AG, Kaiser-Wilhelm-Allee 1
51373 Leverkusen, Germany}

\maketitle

\begin{abstract}
  While classical convolutional neural networks (CNNs) have revolutionized image classification, the emergence of quantum computing presents new opportunities for enhancing neural network architectures. Quantum CNNs (QCNNs) leverage quantum mechanical properties and hold potential to outperform classical approaches. However, their implementation on current noisy intermediate-scale quantum (NISQ) devices remains challenging due to hardware limitations. In our research, we address this challenge by introducing an encoding scheme that significantly reduces the input dimensionality. We demonstrate that a primitive QCNN architecture with 49 qubits is sufficient to directly process $28\times 28$ pixel MNIST images, eliminating the need for classical dimensionality reduction pre-processing. Additionally, we propose an automated framework based on expressibility, entanglement, and complexity characteristics to identify the building blocks of QCNNs, parameterized quantum circuits (PQCs). Our approach demonstrates advantages in accuracy and convergence speed with a similar parameter count compared to both hybrid QCNNs and classical CNNs. We validated our experiments on IBM's Heron r2 quantum processor, achieving $96.08\%$ classification accuracy, surpassing the $71.74\%$ benchmark of traditional approaches under identical training conditions. These results represent one of the first implementations of image classifications on real quantum hardware and validate the potential of quantum computing in this area.
\end{abstract}

\section{Introduction}
Recent advancements in the field of quantum computing have shown promising potential to significantly enhance machine learning through various means, including improved speed, increased memory efficiency, reduced number of parameters required for training, and enhanced privacy measures \cite{abbas2021power, biamonte2017quantum, schuld2015introduction, bauckhage2022quantum}. In the domain of quantum machine learning, there is a hypothesis that quantum systems are capable of generating unique patterns that classical systems find challenging to replicate efficiently \cite{abbas2021power, biamonte2017quantum}. This leads to the expectation that quantum computers may be able to learn and solve problems that are beyond the reach of classical computers and might unlock solutions to some of today's most challenging problems.

Recently, the concept of quantum convolutional neural networks (QCNNs) has been developed \cite{grant2018hierarchical, cong2019quantum}. QCNNs are structurally similar to their classical counterparts but operate within the framework of quantum computing. In QCNNs, parameterized quantum circuits (PQCs) are the main building blocks. PQCs are used to perform convolution-like operations on quantum data. Subsequently, pooling operations can be conducted using specific (parameterized) quantum circuits to reduce the feature map dimensionality. The alternating application of quantum convolution and pooling layers continues until the system's size is reduced sufficiently for a prediction to be made. As in classical CNNs, the prediction accuracy of a QCNN is improved by optimizing the parameters its PQCs. QCNNs are particularly intriguing because they inherently avoid barren plateaus \cite{pesah2021absence} and leverage advantages of quantum machine learning algorithms, including the ability to explore high-dimensional Hilbert spaces and potentially achieve faster convergence, as observed in quantum neural networks \cite{abbas2021power}. 

In recent research, many QCNN architectures have been proposed, including variational quantum circuits with hierarchical structure \cite{grant2018hierarchical, cong2019quantum, hur2022quantum}, hybrid QCNNs \cite{oh2020tutorial, liu2021hybrid, chen2022quantum}, arithmetic-based QCNNs \cite{kerenidis2019quantum, li2020quantum, wei2022quantum}, and QCNNs based on random quantum circuits \cite{henderson2020quanvolutional, yang2021decentralizing}. However, because of the footprint of most QCNNs, even binary MNIST classification problems are out of reach for current NISQ devices. Hybrid models suffer from expensive measurements, and state-of-the-art feature embeddings have either little impact on the qubit count (qubit encoding, QE) or are not feasible on current NISQ devices (amplitude encoding, AE). Therefore, additional dimensionality reduction techniques are often employed to make problems more tractable for quantum circuits, such as PCA, autoencoder, or image resizing \cite{grant2018hierarchical, hur2022quantum, oh2020tutorial}.

While these techniques are common in quantum contexts due to hardware constraints, they are typically avoided in classical CNNs for several reasons. Such techniques can result in the loss of fine details crucial for classification tasks, disrupt the spatial structure that CNNs are designed to exploit and create redundancy with CNNs' inherent dimensional reduction through layer progression. These concerns extend to QCNNs, potentially obscuring whether the QCNN or the preprocessing is doing the classification and limiting generalization to more complex architectures \cite{grant2018hierarchical}. Nevertheless, current quantum hardware limitations necessitate some form of dimensionality reduction or hybrid quantum-classical architectures for QCNN implementations. 

In this work, we propose a hybrid architecture and a QCNN architecture based on a new encoding scheme. Moreover, we introduce a systematic method to design PQCs for variational quantum algorithms such as QCNNs in the convolution layer. While most QCNNs use convolution circuits that are either developed manually or drawn from existing literature \cite{grant2018hierarchical, liu2021hybrid, oh2020tutorial, hur2022quantum, chen2022quantum}, our approach provides automated optimization given desirable properties such as expressibility, entanglement, and size. We analyze the advantages of fully quantum mechanical QCNNs compared to hybrid architectures and optimized classical CNNs, conducting experiments on both binary and multinomial image classifications. Finally, we evaluate our QCNN on IBM's Heron r2 \cite{IBMQuantum} quantum processor released on July 2024 with 156 qubits.

\section{Bayesian optimization}
In this work, Bayesian optimization \cite{NIPS2012_05311655} is employed to address two crucial design challenges: (i) building a minimal yet accurate CNN baseline and (ii) discovering effective parameterized quantum circuit (PQC) ansätze for QCNN convolution operations. Both tasks involve searching large and complex design spaces where brute-force or manual tuning becomes infeasible. For the CNN, we seek to reduce the model's parameter count without sacrificing accuracy, ensuring that the classical counterpart has an 'optimal' architecture. Meanwhile, for the PQC ansatz search, we aim to satisfy rigorous thresholds for expressibility and entanglement while minimizing the overall circuit complexity \textendash{} requirements that directly influence a QCNN’s ability to represent and process quantum data efficiently. The technical details of our Bayesian optimization methodology are provided in \Cref{app:bo_setup}.

\subsection{CNN baseline} 
A natural way to build a CNN baseline might be to mirror the QCNN architecture \textendash{} treating each quantum layer as its classical analog. However, a design that excels in the quantum domain does not necessarily translate into an effective classical CNN. Therefore, given the ongoing controversies around neural architecture search (NAS) and the relatively simple nature of our classification tasks, we employ Bayesian optimization to systematically design the CNN baseline \cite{he2021automl, stamoulis2019single}.

We conduct 50,000 trials of Bayesian optimization to design an optimized CNN architecture that balances minimal parameters with high accuracy, employing the Adam optimizer \cite{kingma2014adam} and ELU activation functions, which offer advantages in convergence speed \cite{clevert2016fastaccuratedeepnetwork, pedamonti2018comparison, wang2020improvement}. A sigmoid activation is used in the output layer. We constrain the models to have fewer than a total number of parameters ($\text{params}_{max}$) and require them to exceed $90\%$ accuracy in 10-fold cross-validation. The objective function is defined as
\begin{equation}
\label{eq:bo_cnn_obj}
    \mathcal{L}_{CNN}=
\begin{cases}
1.9-\text{acc}, & \text{if } \text{acc}\leq\text{0.9}\\
\frac{\text{params}}{\text{params}_{max}}, &  \text{otherwise} \\
\end{cases},
\end{equation}
where accuracy (acc) above the threshold leads to optimizing the parameter count (params). For models with equal parameter counts, the one with higher accuracy is selected. 

All architectures are evaluated on the MNIST dataset \cite{lecun1998mnist}. MNIST is a dataset of handwritten digits from 0 to 9 that is commonly used for evaluating image classification. The dataset contains 70,000 images, each being a grayscale $28\times28$ pixel square. Following prior work \cite{hur2022quantum, grant2018hierarchical}, we use the binary classification task involving only the digits 0 and 1 as our baseline. If the QCNNs demonstrate a significant advantage over classical CNNs in this binary classification task, the model's capabilities are tested on more complex tasks, including 7 vs 8, greater than 4, and multi-class classification of digits 0-3.

To limit the computational cost, we employ models with a single output neuron producing values between 0 and 1 for classification, following an approach introduced for universal quantum classifiers \cite{perez2020data}. In this work, we call this bin-based single output classification (BSOC) problem. For target classes $\mathcal{Y}=\{y_1,y_2,...,y_n\}$, we partition the interval $[0, 1]$ into $n$ equally sized, consecutive bins $\mathcal{B} = \{B_1, B_2, \ldots, B_{n}\}$, where each bin $B_i$ corresponds to class $y_{i}$. Given an input $x$ from the input space $\mathcal{X}$ and model parameters $\theta$, the classifier's output $f_\theta(x)\in B_i$ determines the predicted class $y_i$. For uncertainty estimation, we propose that the model's confidence should be highest at the center $c_i$ of the correct class bin $B_i$, following similar reasoning to margin maximization in support vector machines \cite{cortes1995support}. 

A typical multinomial loss function, such as categorical cross-entropy, can lead to vanishing gradients or local minima for the BSOC problem (\Cref{app:bo_setup}, Proposition \ref{prp:bsoc}). To address this, we adopt a one-vs-rest classification approach with the model's confidence reflecting the distance to the center of the correct class bin, effectively transforming the problem into a regression task. Following the regression approach that was proven effective in Ref. \cite{wang2020regression}, we map the labels to the center of their respective class bins with added normal distributions of small standard deviations $\sigma=0.02$ and apply the MAE loss
\begin{equation}
L(\hat{\mathcal{Y}},f_\theta(\mathcal{X})) = \frac{1}{m}\sum_{i=1}^{m}|\hat{y}_i-f_\theta(x_i)|,
\end{equation}
where $\hat{y}_i\in\hat{\mathcal{Y}}$ are the projected and normal distributed labels and $m$ denotes the number of samples. Based on modifications of Ref. \cite{wang2020regression}, a model with 100,000 parameters is sufficient for classifying the entire MNIST dataset with this approach, allowing us to set $\text{params}_{max}$ in \Cref{eq:bo_cnn_obj} accordingly.

The resulting architecture for the 0-vs-1 classification task is presented in \Cref{tab:network_architecture_cl01}, with architectures for the other classification tasks (7-vs-8, greater than 4, and digits 0-3) detailed in Section SII of the Supplemental Material. While this approach is not intended as a replacement for more advanced NAS heuristics \cite{albelwi2017framework, he2021automl, stamoulis2019single, lyu2024survey}, it serves as a practical method for this simple problem setting. 

\begin{table}[t]
    \centering
    \renewcommand{\arraystretch}{1.5} 
    \setlength{\tabcolsep}{12pt} 
    \begin{tabular}{c c c} 
    \hline
    Layer & Filter Size & \makecell{Activation\\Function} \\ 
    \hline
    Conv2D-1 & $3\times3\times1$ & ELU \\ 
    MaxP2D-1 & $2\times2$ & - \\ 
    Conv2D-2 & $1\times1\times1$ & ELU \\ 
    MaxP2D-2 & $2\times2$ & - \\ 
    Conv2D-3 & $2\times2\times1$ & ELU \\ 
    MaxP2D-3 & $2\times2$ & - \\ 
    Conv2D-4 & $2\times2\times1$ & ELU \\ 
    MaxP2D-4 & $2\times2$ & Sigmoid \\ 
    \hline
    \end{tabular}
    \caption{Optimized CNN architecture for classes 0 and 1 with $22$ parameters.}
    \label{tab:network_architecture_cl01}
\end{table}

\subsection{Ansatz search} Similar to our CNN architecture optimization, we employ Bayesian optimization to design quantum circuits for the QCNN convolution operations. We construct parameterized quantum circuits (PQCs) using elementary gates described in Table S4 of the Supplemental Material. The PQCs are evaluated based on being capable of exploring the solution space (expressibility), incorporating all encoded information (entanglement), and maintaining a reasonable resource size (costs) \cite{sim2019expressibility}.

To assess expressibility, the output distribution of a circuit is compared to a target distribution $P_{target}(\Omega)$ over a set $C$ of different state initializations. Different initializations are essential to evaluate the expressibility of a quantum circuit. For example, a single qubit with an $R_z$ gate explores in state $\ket{0}$ only a single point while in state $\ket{+}$ the entire equator of the Bloch sphere. The output distribution is formed by sampling from the PQC with a set $S$ of circuit parameters. Using Kullback-Leibler (KL) divergence \cite{kullback1951information}, we measure the deviation as 
\begin{equation}
\text{expr} = \frac{1}{|C|}\sum_{c\in C}D_{KL}(\hat{P}_{PQC}(\Omega;S,c)||P_{target}(\Omega)), 
\end{equation}
where $\hat{P}_{PQC}(\Omega;S,c)$ is the output distribution of the PQC. The expressibility loss function is then given by
\begin{equation}
    \mathcal{L}_{expr} = \max\left( \frac{\text{expr} - \text{expr}_{thr}}{\text{expr}_{max} - \text{expr}_{thr}}, 0 \right),
\end{equation}
where $\text{expr}_{thr}$ is a set threshold to reach and $\text{expr}_{max}$ is the maximum possible divergence. A high $\mathcal{L}_{expr}$ value indicates that the circuit's output distribution significantly deviates from $P_{target}(\Omega)$, suggesting a limited ability to explore the solution space.

For entanglement evaluation, we employ the Meyer-Wallach metric \cite{meyer2002global} for quantum states $\ket{\psi^c_\theta}$ sampled from the PQC with initial state c and parameters $\theta$ by 
\begin{equation} 
    \text{entgl}=\frac{1}{|C|\cdot|S|}\sum_{c\in C}\sum_{\theta\in S}Q(\ket{\psi^c_{\theta}}),
\end{equation}
similar to Ref. \cite{sim2019expressibility}. Analogous to the expressibility loss, we define the entanglement loss function
\begin{equation}
    \mathcal{L}_{entgl} = \max\left( \frac{\text{entgl}_{thr} - \text{entgl}}{\text{entgl}_{thr}}, 0 \right),
\end{equation}
where $\text{entgl}_{thr}$ is the entanglement threshold. As $0\leq \text{entgl}\leq1$ \cite{brennen2003observable}, no additional normalization is required.

For the cost of a PQC, we evaluate the quantum circuit's complexity following Ref. \cite{sim2019expressibility}, using the number of parameters (params) and circuit depth (depth). However, instead of considering connectivity, we count the total number of gates (gates) since convolution operations inherently require full qubit connectivity. The complexity loss function is given by
\begin{equation}
    \mathcal{L}_{cmplx} = \frac{\text{gates} + \text{params} + \text{depth}}{\text{gates}_{max}+\text{params}_{max}+\text{depth}_{max}},
\end{equation}
where $\text{gates}_{max},\ \text{params}_{max},$ and $\text{depth}_{max}$ are the maximum allowed circuit complexity. 

Finally, the objective function $\mathcal{L}_{PQC}$ for the Bayesian optimization is built as
\begin{equation}
    \mathcal{L}_{PQC} = 
    \begin{cases} 
    \mathcal{L}_{expr} + \mathcal{L}_{entgl} + 1 & \text{if } \mathcal{L}_{expr} + \mathcal{L}_{entgl} \neq 0 \\
    \mathcal{L}_{cmplx}  & \text{otherwise}
    \end{cases}.
\end{equation}
$\mathcal{L}_{PQC}$ prioritizes meeting the expressibility and entanglement thresholds. If both thresholds are met ($\mathcal{L}_{expr} + \mathcal{L}_{entgl} = 0$), the objective function focuses on minimizing circuit complexity through $\mathcal{L}_{cmplx}$.

\section{QCNN architectures} 
We investigate the impact of quantum convolutional neural networks (QCNNs) on classification accuracy by comparing two distinct architectures: a hybrid QCNN and a QCNN utilizing fragment encoding. The hybrid QCNN performs measurements after each kernel application, thereby reducing the required qubit count to depend solely on the kernel size, effectively addressing current quantum hardware limitations \cite{oh2020tutorial, chen2022quantum}. In contrast, the proposed QCNN utilizing fragment encoding analyzes the entire input image with a single measurement by dividing the image into fragments, each encoded in parallel using a sequence of single-qubit gates. All setup details of the classification are provided in \Cref{app:clf_setup}.

\subsection{Hybrid QCNN}
To identify convolution circuits for the hybrid QCNN, expressibility ensures unbiased mapping between random inputs and class labels. This is assessed by comparing the circuit's output distribution to a uniform distribution over the target classes $P_{U}(\mathcal{Y})$. Given a set $C$ of uniform random input segments and a set $S$ of random parameters, the deviation is measured by
\begin{equation}
    \text{expr} = \frac{1}{|C|}\sum_{x\in C}D_{KL}(\hat{P}_{PQC}(\mathcal{Y}; S, x)||P_{U}(\mathcal{Y})),    
\end{equation}
where $\hat{P}_{PQC}(\mathcal{Y}; S, x)$ is the probability distribution over the target classes obtained by measuring the last qubit from the PQC for parameters $S$ given an input $x$. A higher KL divergence value indicates a greater deviation from the uniform distribution, suggesting that the circuit has lower expressibility and may be biased toward certain class labels. This concept extends to the entire hybrid QCNN through sequences of expressible circuits. 

The classification of the QCNN is based on the following steps:
\begin{enumerate}
    \item Encode image segments (see \Cref{app:encoding})
    \item Apply convolution/pooling
    \item Measure last qubit
    \item Repeat steps (1)-(3) for each layer 
\end{enumerate}
This allows the architecture to efficiently process large inputs in parallel with the required qubits only depending on the kernel size (\Cref{fig:QCNN-Meas-Mechanism}). After each training batch, the network is updated by adjusting the quantum circuit's parameters.

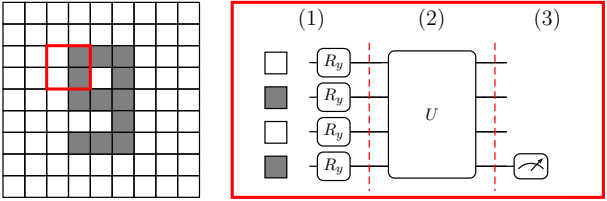
\begin{figure}[t]
        \centering
        \resizebox{0.475\textwidth}{!}{
        \begin{tikzpicture}
            [
                box/.style={rectangle,draw=black,thick, minimum size=0.5cm},
            ]

        \foreach \x in {0,0.5,...,4}{
            \foreach \y in {0,0.5,...,4}
                \node[box] at (\x,\y){};
        }
    
        \node[box,fill=gray] at (2.5,2.5){};
        \node[box,fill=gray] at (2.5,3){};
        \node[box,fill=gray] at (2,3){};
        \node[box,fill=gray] at (1.5,3){};
        \node[box,fill=gray] at (1.5,2.5){};
        \node[box,fill=gray] at (1.5,2){};
        \node[box,fill=gray] at (2,2){};
        \node[box,fill=gray] at (2.5,2){};
        \node[box,fill=gray] at (2.5,1.5){};
        \node[box,fill=gray] at (2.5,1){};
        \node[box,fill=gray] at (2,1){};
        \node[box,fill=gray] at (1.5,1){};
    
        \draw[draw=red, line width=2pt] (0.75,2.25) rectangle (1.75,3.25);
        \draw[draw=red, line width=2pt] (5,-0.25) rectangle (13.5,4.25);
        
        \node at (9.5,1.75) {
            \begin{quantikz}[rounded corners, column sep=5pt, row sep=4pt]
                & \gate{R_y} & \qw & \qw\slice{} & \qw & \qw & \gate[4][2cm]{U} & \qw & \qw \slice{} & \qw  & \qw \\
                & \gate{R_y} & \qw & \qw & \qw & \qw & \qw & \qw & \qw & \qw & \qw \\
                & \gate{R_y} & \qw & \qw & \qw& \qw & \qw & \qw & \qw & \qw & \qw \\
                & \gate{R_y} & \qw & \qw & \qw & \qw & \qw & \qw & \qw & \qw & \qw & \meter{} 
            \end{quantikz}
        };
    
        \node[box,fill=gray] at (6, 0.45){};
        \node[box] at (6,1.25){};
        \node[box,fill=gray] at (6,2.05){};
        \node[box] at (6,2.85){};
        
        \node[text width=3cm] at (8,3.85) 
        {\Large (1)};
        \node[text width=3cm] at (10.75,3.85) 
        {\Large (2)};
        \node[text width=1.2cm] at (12.5,3.85) 
        {\Large (3)};

        \end{tikzpicture}
        }
        \caption{Illustration of the hybrid QCNN Type II quantum circuit used for the kernel operation. The 
        input fragments are (1) encoded, (2) processed, and (3) the last qubit is measured.}
        \label{fig:QCNN-Meas-Mechanism}
    \end{figure}

While feasible on current NISQ devices, this hybrid architecture requires costly repeated measurements and sacrifices global entanglement effects that may be crucial for QCNN advantages \cite{hur2022quantum}. The model must also ensure distinguishability between different inputs at each layer \textendash{} a condition not guaranteed for all embeddings. For instance, a single $R_z$ rotation for encoding followed by an $R_y$ rotation for a $1 \times 1$ convolution produces identical outputs regardless of the input and initial state.

\subsection{Hybrid QCNN results}
Two convolutional designs are proposed for the hybrid QCNN. In Type I, the PQC is initially applied to each row of the input mask, followed by its application to the final column (\Cref{fig:mqcnn_type1_conv}). Upon completion of these operations, the last qubit is measured. This architecture exhibits a tree-like structure, with the hypothesis that if the circuit successfully convolves information into the last qubit during the first step, it will similarly do so in the subsequent step.

In Type II, the PQC is applied to all input qubits simultaneously, after which the state of the last qubit is measured (\Cref{fig:QCNN-Meas-Mechanism}). This architecture allows for a greater degree of freedom in how the qubits are connected. Both architectures were tested with layers constructed analogously to the classical CNN baseline using qubit encoding (QE); see \Cref{app:encoding}.

\begin{figure}[t]
    \centering
        \resizebox{0.495\textwidth}{!}{
    \begin{tikzpicture}
        [
            box/.style={rectangle,draw=black,thick, minimum size=0.5cm},
        ]

    \node[box] (a) at (-4,-0.25) {};
    \node[box] (b) at (-3.5,-0.25) {};
    \node[box, fill=cyan] (c) at (-4,0.25) {};
    \node[box, fill=cyan] (d) at (-3.5,0.25) {};
    \node at (-3.75,-3) {(1)};

    \node[box] (e) at (-4,-1.75) {};
    \node[box] (f) at (-3.5,-1.75) {};
    \node[box, fill=cyan] (g) at (-4,-2.25) {};
    \node[box, fill=cyan] (h) at (-3.5,-2.25) {};
    
    \node[box] (i) at (-2,-0.75) {};
    \node[box, fill=cyan] (j) at (-1.5,-0.75) {};
    \node[box] (k) at (-2,-1.25) {};
    \node[box, fill=cyan] (l) at (-1.5,-1.25) {};
    \node at (-1.75,-3) {(2)};

    \node at (0,-1) {$=$};
    
    \node at (4,-1) {
        \begin{quantikz}[rounded corners, column sep=5pt, row sep=4pt]
            & \gate{R_y} & \qw & \qw & \qw & \qw & \gate[2][1cm]{U}\gateinput[1]{0} & \qw & \qw  & \qw  & \qw & \qw \\
            & \gate{R_y} & \qw & \qw & \qw & \qw & \qw\gateinput[1]{1} & \qw & \qw & \gate[3][1cm]{U}\gateinput[1]{0} & \qw & \qw \\
            & \gate{R_y} & \qw & \qw & \qw& \qw & \gate[2][1cm]{U}\gateinput[1]{0} & \qw & \qw & \qw & \qw & \qw \\
            & \gate{R_y} & \qw & \qw & \qw & \qw & \qw\gateinput[1]{1} & \qw & \qw & \qw\gateinput[1]{1} & \qw & \qw & \meter{} 
        \end{quantikz}
    };
    \node at (3.5,-3) {(1)};
    \node at (5,-3) {(2)};

    \end{tikzpicture}
    }
    \caption{Hybrid QCNN Type I hierarchical convolution. First, the receptive field is processed row-wise (1), followed by processing the final column (2).}
    \label{fig:mqcnn_type1_conv}
\end{figure}
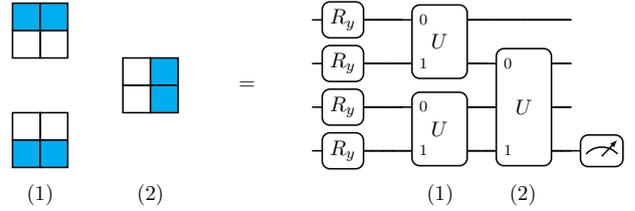

The pooling layers employ the hierarchical circuit of Ref.~\cite{hur2022quantum} (see \Cref{fig:mqcnn_type1_conv}), defined by
\begin{equation}
\label{eq:pooling}
U = (X \otimes I)\,CRX(\theta_0)\,(X \otimes I)\,CRZ(\theta_1),
\end{equation}
illustrated in \Cref{app:qcnn_setup}, \Cref{fig:pooling_circuit}. The circuits for the convolution were based on the circuits discovered in the ansatz search (PQC-Opt) and circuits adapted from previous research (\Cref{fig:convolutional_circuits_2q}), scaled to 2, 3, 4, and 9 qubit configurations (Supplemental Material, Section SIII). The selection of the circuits was based on requiring a similar number of parameters to, or fewer than, the corresponding kernel in a classical CNN. Circuits 1, 5, and 6 are from Ref. \cite{hur2022quantum}, circuit 2 from Ref. \cite{grant2018hierarchical}, and circuits 3 and 4 from Ref. \cite{sim2019expressibility}. Circuit 6 was only employed for 2 and 3 qubits convolution circuits due to the increase in parameters. Since testing all combinations of circuits from previous research is computationally difficult, we only included the best-performing circuits from prior research \textendash{} selected based on expressibility (Exp-Opt), entanglement (Ent-Opt), and the objective function from the ansatz search (Obj-Opt); see \Cref{app:qcnn_setup}, \Cref{tab:mqcnn_metrics_recap}.


\begin{figure*}[!t]
\centering
    \subfloat{
    \begin{minipage}{0.32\textwidth}
        \centering
        \begin{quantikz}[rounded corners, column sep=5pt, row sep=4pt]
        & \gate[][0em][1.75em]{H} & \ctrl{1} & \gate{R_x(\theta_0)} & \qw \\
        & \gate[][0em][1.75em]{H} & \ctrl{0} & \gate{R_x(\theta_1)} & \qw
        \end{quantikz}
        
        \vspace{0.5em} 
  
        {\small Circuit 1}
    \end{minipage}
    }
    \subfloat{
    \begin{minipage}{0.32\textwidth}
        \centering
        \begin{quantikz}[rounded corners, column sep=5pt, row sep=4pt]
        & \gate{R_y(\theta_0)} & \ctrl{1} & \qw & \qw\\
        & \gate{R_y(\theta_1)} & \targ{} & \gate{R_y(\theta_2)} & \qw
        \end{quantikz}
        
        \vspace{0.5em} 
  
        {\small Circuit 2}
    \end{minipage}
    }
    \subfloat{
    \begin{minipage}{0.32\textwidth}
        \centering
        \begin{quantikz}[rounded corners, column sep=5pt, row sep=4pt]
        & \gate{R_x(\theta_0)} & \gate{R_z(\theta_2)} & \targ{} & \qw \\
        & \gate{R_x(\theta_1)} & \gate{R_z(\theta_3)} & \ctrl{-1} & \qw
        \end{quantikz}
        
        \vspace{0.5em} 
  
        {\small Circuit 3}
    \end{minipage}
    }
    
    \vspace{0.25cm}
    \subfloat{
    \begin{minipage}{0.48\textwidth}
        \centering
        \begin{quantikz}[rounded corners, column sep=5pt, row sep=4pt]
        & \gate{R_y(\theta_0)} & \ctrl{1} & \gate{R_y(\theta_2)} & \qw \\
        & \gate{R_y(\theta_1)} & \ctrl{-1} & \gate{R_y(\theta_3)} & \qw
        \end{quantikz}

        \vspace{0.5em} 
  
        {\small Circuit 4}
        
    \end{minipage}
    }
    \hspace{-0.2cm} 
    \subfloat{
    \begin{minipage}{0.48\textwidth}
        \centering
        \begin{quantikz}[rounded corners, column sep=5pt, row sep=4pt]
        & \gate{R_y(\theta_0)} & \targ{} & \gate{R_y(\theta_2)} & \ctrl{1} & \qw \\
        & \gate{R_y(\theta_1)} & \ctrl{-1} & \gate{R_y(\theta_3)} & \targ{} & \qw
        \end{quantikz}

        \vspace{0.5em} 
  
        {\small Circuit 5}
    \end{minipage}
    }
    \vspace{0.25cm}

    \subfloat{
    \begin{minipage}{0.4\textwidth}
    \centering
    \begin{quantikz}[rounded corners, column sep=5pt, row sep=4pt]
    & \gate{R_y(\theta_0)} & \gate{R_z(\theta_2)} & \gate{R_y(\theta_3)} &  \ctrl{1} & \qw \\
    & \gate{R_y(\theta_1)} & \ctrl{-1} & \gate{R_y(\theta_4)} & \gate{R_z(\theta_5)} & \qw
    \end{quantikz}

    \vspace{0.5em} 
  
    {\small Circuit 6}
    \end{minipage}
    }
\hspace{-0.2cm} 
    \subfloat{
    \begin{minipage}{0.4\textwidth}
    \centering
    \begin{quantikz}[rounded corners, column sep=5pt, row sep=4pt]
        & \gate{Z} & \ctrl{1} & \gate{R_x(\theta_0)} & \qw & \qw \\
        & \gate{R_y(\theta_0)} & \ctrl{-1} & \ctrl{-1} & \gate{R_y(\theta_1)} & \qw
    \end{quantikz}
    
    \vspace{0.5em} 
  
    {\small Hybrid ansatz search}
\end{minipage}
    }

    \vspace{0.25cm}
    \subfloat{
    \begin{minipage}{0.97\textwidth}
        \centering
        \begin{quantikz}[rounded corners, column sep=5pt, row sep=4pt]
        & \gate{R_z(\theta_0)} & \ctrl{1} & \gate{U3(\theta_0,\theta_0,\theta_0)} & \qw & \gate[2]{ECR} \gateinput{1} & \gate{U3(\theta_1,\theta_0,\theta_0)} & \qw \\
        & \gate{Y} & \gate{R_z(\theta_0)} & \gate{R_y(\theta_0)} & \gate{R_y(\theta_0)} & \gateinput{0} & \qw & \qw 
        \end{quantikz}
        
        \vspace{0.5em} 
  
        {\small Regular ansatz search}
    \end{minipage}
    }

\caption{Parameterized quantum circuits for a 2-qubit convolution. Circuits 1 to 6 are adapted from previous research papers. Circuit 7 and 8 are from the ansatz search.}
\label{fig:convolutional_circuits_2q}
\end{figure*}
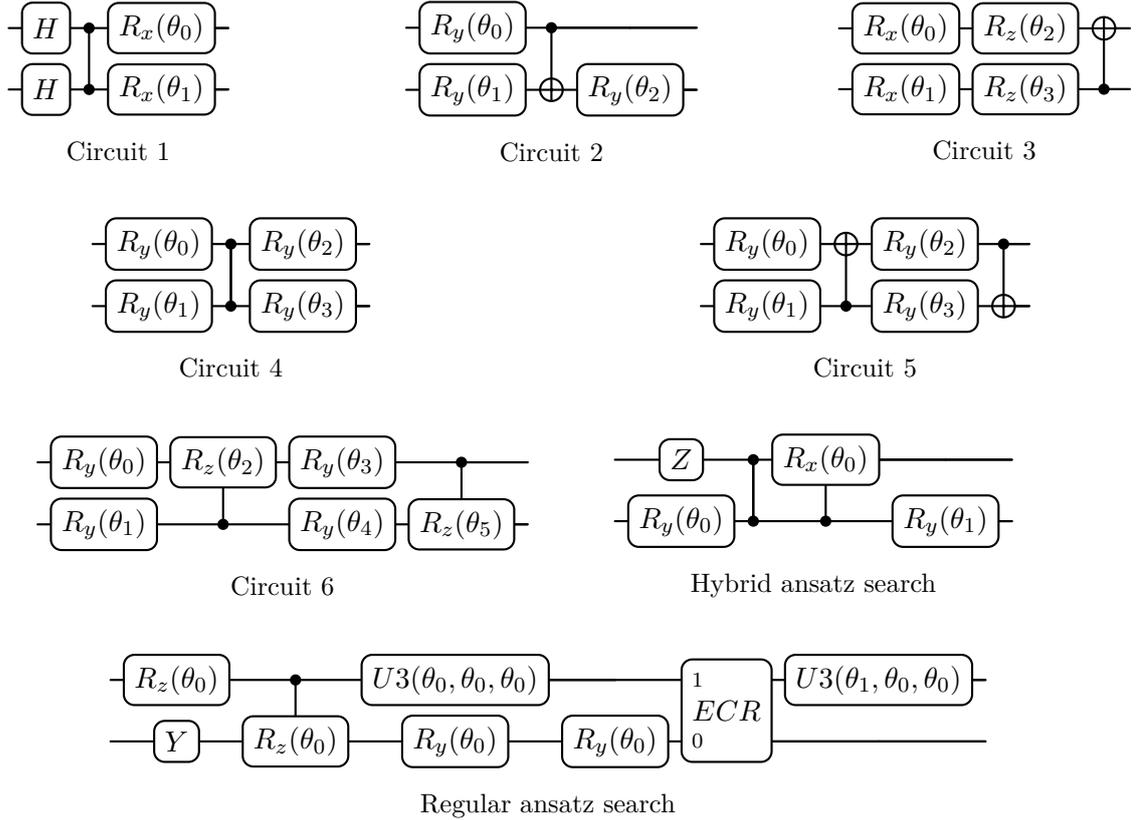

In \Cref{fig:ansatz_search_res}, circuits identified through ansatz search consistently demonstrated superior performance in the overall objective compared to those from previous research. As the qubit count increases, the ansatz search continues to find circuits that maintain good expressibility and entanglement characteristics, while circuits from previous research show degraded performance. However, with the exception of circuit 4, most ansatz search circuits did not meet the expressibility and entanglement thresholds. This likely stems from the difference in sample sizes between the Bayesian optimization evaluation and the final evaluation (\Cref{app:bo_setup_as}). During Bayesian optimization, the circuits appeared to pass the thresholds, which also allowed for optimization of circuit size, resulting in the notably more compact designs detailed in Section SIV of the Supplemental Material. When these same circuits underwent final evaluation with a larger sample size, they failed to meet the established thresholds.

For the classification of digits 0 and 1, the hybrid QCNN achieved results comparable to or slightly below the CNN, as shown in \Cref{fig:mqcnn_results}. For Type I hybrid QCNN implementations, the Obj-Opt model achieved the highest accuracy among quantum variants, while the newly discovered PQC demonstrated lower accuracy but used fewer parameters than the classical model. Type II hybrid QCNNs achieved stronger results, with the expressibility-optimized model reaching $92.29\%$ accuracy, approaching CNN performance with lower variance. The model with newly discovered PQCs achieves $87.47\%$ accuracy with a much lower parameter count compared to the other Type II hybrid QCNNs.

Comparative analysis across all architectures revealed consistent correlations between the metric values of the convolution circuits and classification performance. High entanglement consistently corresponded with improved accuracy, while expressibility showed varying correlation patterns. The value of the objective function proved to be an inconsistent predictor of performance, particularly in Type II implementations. Notably, in our experimental observations, all quantum models showed slower or at best equal convergence speed compared to the CNN. All in all, while demonstrating potential, neither hybrid QCNN architecture surpassed the performance of the classical CNN. These findings suggest that further optimization of quantum circuits and architectural design may be necessary before the hybrid QCNN can effectively compete with optimized CNNs in practical applications.

\begin{figure}[t]
\centering
\resizebox{0.495\textwidth}{!}{
\begin{tikzpicture}
\begin{axis}[
    width=14cm,
    height=9cm,
    xlabel={},
    ylabel={$\mathcal{L}_{PQC}$},
    symbolic x coords={0, 2-Hybrid, 2-Regular, 3-Hybrid, 3-Regular, 4-Hybrid, 4-Regular, 9-Hybrid, 9-Regular},
    xtick={},
    xticklabels={},
    ytick={0,0.2,0.4,0.6,0.8,1,1.2,1.4,1.6,1.8,2},
    ymin=0, ymax=2,
    legend style={at={(1.11,0.98)}, anchor=north east, draw=black, fill=white, align=left},
    grid=both,
    clip=false,
    cycle list name=OkabeIto,
    font=\Large,                    
    tick label style={font=\Large}, 
    label style={font=\large},      
]

\addplot+[only marks, mark=o, mark size=3pt, mark options={line width = 1.5pt}] coordinates {
    (2-Hybrid, 1.664) (2-Regular, 1.055) (3-Hybrid, 1.681) (3-Regular, 1.019) 
    (4-Hybrid, 1.784) (4-Regular, 1.008) (9-Hybrid, 1.798) 
};
\addlegendentry{C1};

\addplot+[only marks, mark=square, mark size=3pt, line width = 1.5pt] coordinates {
    (2-Hybrid, 1.376) (2-Regular, 1.378) (3-Hybrid, 1.438) (3-Regular, 1.439) 
    (4-Hybrid, 1.558) (4-Regular, 1.554) (9-Hybrid, 1.7) (9-Regular, 1.562)
};
\addlegendentry{C2};

\addplot+[only marks, mark=triangle, mark size=3pt, line width = 1.5pt] coordinates {
    (2-Hybrid, 1.23) (2-Regular, 1.073) (3-Hybrid, 1.255) (3-Regular, 1.203) 
    (4-Hybrid, 1.348) (4-Regular, 1.246) (9-Hybrid, 1.376) (9-Regular, 1.203)
};
\addlegendentry{C3};

\addplot+[only marks, mark=diamond, mark size=3pt, line width = 1.5pt] coordinates {
    (2-Hybrid, 1.375) (2-Regular, 1.373) (3-Hybrid, 1.438) (3-Regular, 1.439) 
    (4-Hybrid, 1.545) (4-Regular, 1.549) (9-Hybrid, 1.623) (9-Regular, 1.621)
};
\addlegendentry{C4};

\addplot+[only marks, mark=pentagon, mark size=3pt, line width = 1.5pt] coordinates {
    (2-Hybrid, 1.218) (2-Regular, 1.206) (3-Hybrid, 1.234) (3-Regular, 1.063) 
    (4-Hybrid, 1.379) (4-Regular, 1.134) (9-Hybrid, 1.528) (9-Regular, 1.234)
};
\addlegendentry{C5};

\addplot+[only marks, mark=star, mark size=3pt, line width = 1.5pt] coordinates {
    (2-Hybrid, 1.466) (2-Regular, 1.471) (3-Hybrid, 1.405) (3-Regular, 1.407)  
};
\addlegendentry{C6};

\addplot+[only marks, mark=otimes, mark size=3pt, line width = 1.5pt] coordinates {
    (2-Hybrid, 1.028) (2-Regular, 1.006) (3-Hybrid, 1.069) (3-Regular, 1.002) 
    (4-Hybrid, 0.37) (4-Regular, 1.0) (9-Hybrid, 1.032) (9-Regular, 1.0)
};
\addlegendentry{AS};

\draw[dashed, gray, xshift=20pt] (axis cs:2-Regular,0) -- (axis cs:2-Regular,2);
\draw[dashed, gray, xshift=20pt] (axis cs:3-Regular,0) -- (axis cs:3-Regular,2);
\draw[dashed, gray, xshift=20pt] (axis cs:4-Regular,0) -- (axis cs:4-Regular,2);

\coordinate (mid2H) at (axis cs:2-Hybrid,0);
\coordinate (mid2R) at (axis cs:2-Regular,0);
\coordinate (mid3H) at (axis cs:3-Hybrid,0);
\coordinate (mid3R) at (axis cs:3-Regular,0);
\coordinate (mid4H) at (axis cs:4-Hybrid,0);
\coordinate (mid4R) at (axis cs:4-Regular,0);
\coordinate (mid9H) at (axis cs:9-Hybrid,0);
\coordinate (mid9R) at (axis cs:9-Regular,0);

\pgfmathsetmacro{\halfwidth}{0.5}

\fill[midred!45, opacity=0.3] ([xshift=-\halfwidth cm]mid2H) rectangle ([xshift=\halfwidth cm]mid2H |- {axis cs:0,2});
\fill[midgold!45, opacity=0.3] ([xshift=-\halfwidth cm]mid2R) rectangle ([xshift=\halfwidth cm]mid2R |- {axis cs:0,2});
\fill[midred!45, opacity=0.3] ([xshift=-\halfwidth cm]mid3H) rectangle ([xshift=\halfwidth cm]mid3H |- {axis cs:0,2});
\fill[midgold!45, opacity=0.3] ([xshift=-\halfwidth cm]mid3R) rectangle ([xshift=\halfwidth cm]mid3R |- {axis cs:0,2});
\fill[midred!45, opacity=0.3] ([xshift=-\halfwidth cm]mid4H) rectangle ([xshift=\halfwidth cm]mid4H |- {axis cs:0,2});
\fill[midgold!45, opacity=0.3] ([xshift=-\halfwidth cm]mid4R) rectangle ([xshift=\halfwidth cm]mid4R |- {axis cs:0,2});
\fill[midred!45, opacity=0.3] ([xshift=-\halfwidth cm]mid9H) rectangle ([xshift=\halfwidth cm]mid9H |- {axis cs:0,2});
\fill[midgold!45, opacity=0.3] ([xshift=-\halfwidth cm]mid9R) rectangle ([xshift=\halfwidth cm]mid9R |- {axis cs:0,2});

\node[anchor=north, rotate=0] at (axis cs:2-Hybrid,-0.03) {H};
\node[anchor=north, rotate=0] at (axis cs:2-Regular,-0.03) {R};
\node[anchor=north, rotate=0] at (axis cs:3-Hybrid,-0.03) {H};
\node[anchor=north, rotate=0] at (axis cs:3-Regular,-0.03) {R};
\node[anchor=north, rotate=0] at (axis cs:4-Hybrid,-0.03) {H};
\node[anchor=north, rotate=0] at (axis cs:4-Regular,-0.03) {R};
\node[anchor=north, rotate=0] at (axis cs:9-Hybrid,-0.03) {H};
\node[anchor=north, rotate=0] at (axis cs:9-Regular,-0.03) {R};

\node[anchor=north, xshift=20pt] at (axis cs:2-Hybrid,-0.25) {2 Qubits};
\node[anchor=north, xshift=20pt] at (axis cs:3-Hybrid,-0.25) {3 Qubits};
\node[anchor=north, xshift=20pt] at (axis cs:4-Hybrid,-0.25) {4 Qubits};
\node[anchor=north, xshift=20pt] at (axis cs:9-Hybrid,-0.25) {9 Qubits};

\end{axis}
\end{tikzpicture}
}
\caption{Comparison of $\mathcal{L}_{PQC}$ between the 6 quantum circuits (C) from prior studies and circuits identified through ansatz search (AS) across 4 qubit configurations. For each qubit count, both hybrid QCNN (H) and regular QCNN (R) implementations are shown with distinct background colors.}
\label{fig:ansatz_search_res}
\end{figure}
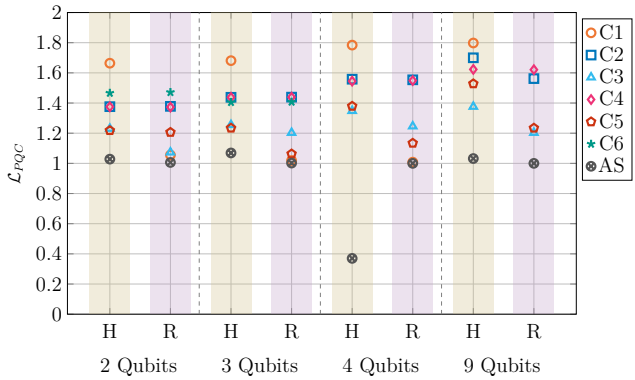

\begin{figure}[t]
    \centering
    \includegraphics[width=0.99\linewidth]{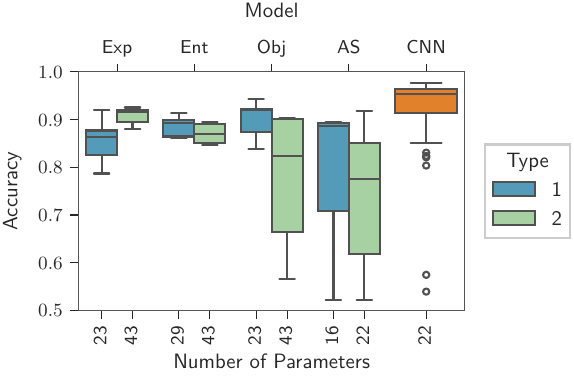}
    \caption{Hybrid QCNN Type I (first row) and II (second row) vs CNN baseline for $\{0,1\}$ classification. Comparing models optimized for expressibility (Exp-Opt), entanglement (Ent-Opt), objective function (Obj-Opt), and newly discovered PQC (PQC-Opt).
    }
    \label{fig:mqcnn_results}
\end{figure}

\subsection{Regular QCNN} 
The main challenges from non-hybrid QCNNs \cite{grant2018hierarchical, cong2019quantum, hur2022quantum}, stem from either the excessive number of qubits required for encoding, such as with QE, or the substantial depth of the quantum circuit, as for amplitude encoding \cite{grant2018hierarchical, hur2022quantum}. We introduce a feasible quantum mechanical encoding for existing quantum hardware without costly measurements or dimensionality reduction techniques, the fragment encoding. 

The fragment encoding mimics convolution layers, reducing the input to a size that can be efficiently processed by a fully quantum mechanical QCNN. An element in the $i$-th hidden layer combines portions of the previous layer, which themselves are derived from earlier layers. We exploit this structure by computing only the necessary fragments from preceding layers, effectively dividing the input image into fragments for parallel processing (\Cref{fig:conv_layers}). This approach solves the encoding problem for QCNNs if each image fragment can be encoded by a single qubit. 

\begin{figure}[t]
    \centering
    \subfloat[Classical convolution layers\label{fig:conv_layers}]{
    \begin{minipage}{0.43\textwidth}
    \centering
    \begin{tikzpicture}
        \begin{scope}[cm={1.5,1,0,2,(0,0)}]
        \node[transform shape, draw, gray, ultra thick, inner sep=0.2mm] {\includegraphics[width=.25\textwidth]{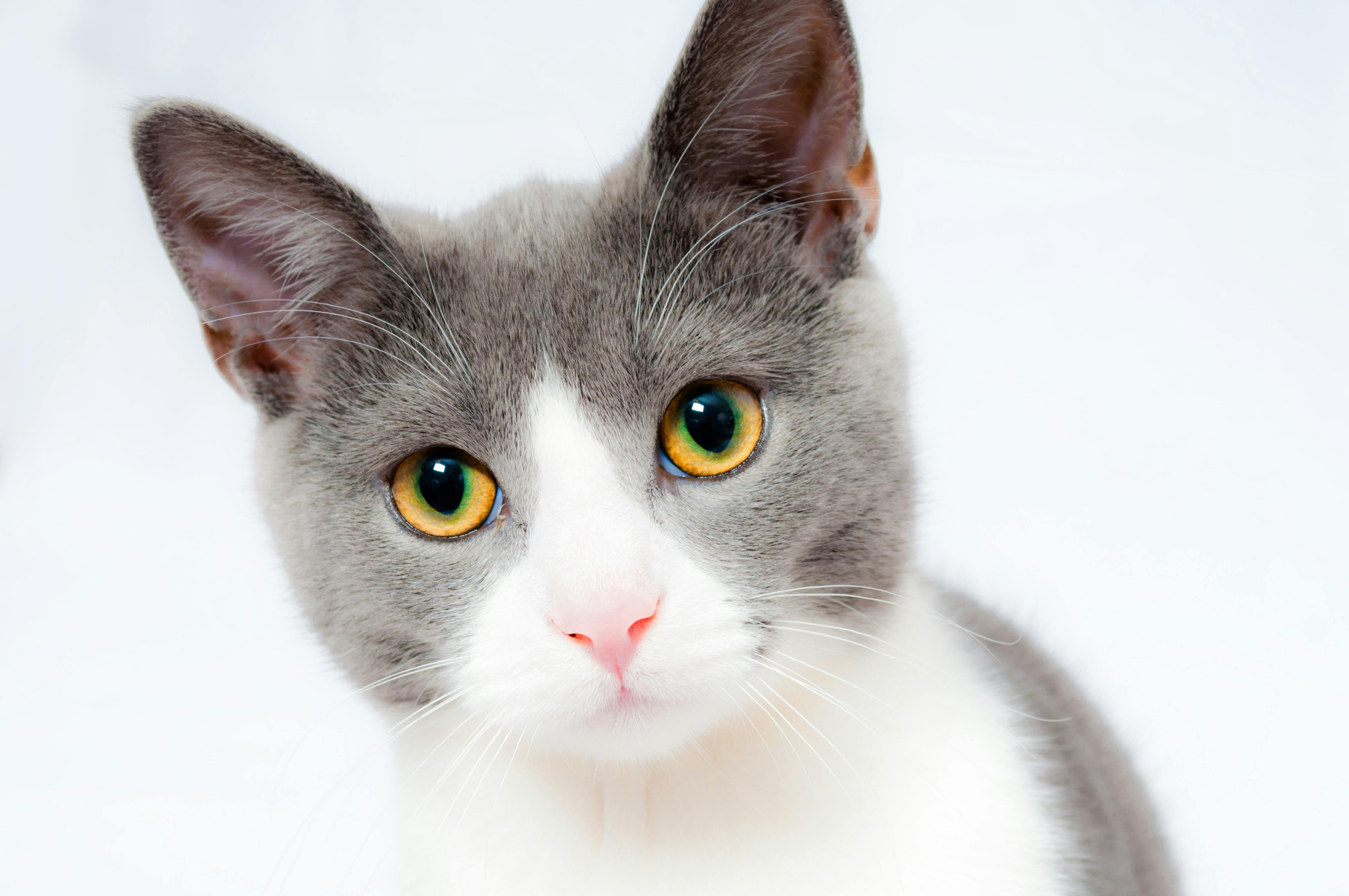}};

        \coordinate (A0) at (0.2,0);
        \coordinate (B0) at (0.2,0.4);
        \coordinate (C0) at (0.6,0);
        \coordinate (D0) at (0.6,0.4);
        \draw[red, dotted, xstep=0.1cm, ystep=0.1cm] (A0) grid (D0);
                
        \draw[draw=blue] (1.3,-1.8) rectangle ++(1.2,0.8);
        \coordinate (A1) at (2.2,-1.4);
        \coordinate (B1) at (2.2,-1.2);
        \coordinate (C1) at (2.4,-1.4);
        \coordinate (D1) at (2.4,-1.2);
        \draw[red, line, xstep=0.1cm, ystep=0.1cm] ($(A1)-(0,0.001cm)$) grid (D1);

        \draw[dotted, line width=0.5, color=red] (A0)--(A1);
        \draw[dotted, line width=0.5, color=red] (B0)--(B1);
        \draw[dotted, line width=0.5, color=red] (C0)--(C1);
        \draw[dotted, line width=0.5, color=red] (D0)--(D1);

        \draw[draw=blue] (3.5,-3) rectangle ++(0.75,0.5);
        \coordinate (A2) at (4,-2.7);
        \coordinate (B2) at (4,-2.6);
        \coordinate (C2) at (4.1,-2.7);
        \coordinate (D2) at (4.1,-2.6);
        \draw[red, xstep=0.1cm, ystep=0.1cm] ($(A2)-(0,0.001cm)$) grid ($(D2)+(0.001cm,0)$);

        \draw[dotted, line width=0.5, color=red] (A1)--(A2);
        \draw[dotted, line width=0.5, color=red] (B1)--(B2);
        \draw[dotted, line width=0.5, color=red] (C1)--(C2);
        \draw[dotted, line width=0.5, color=red] (D1)--(D2);

        
        \end{scope}

    \end{tikzpicture}
    \end{minipage}
    }
    \vspace{1cm}
    \hfill
    \subfloat[Fragment encoding\label{fig:fragment_enc}]{
    \begin{minipage}{0.35\textwidth}
        \centering
        \resizebox{1.5\textwidth}{!}{
        \begin{tikzpicture}[shift={(-2cm,0)}]
            \useasboundingbox (15, -2) rectangle (1, 0);
            \tikzstyle{box} = [rectangle,draw=black,thick, minimum size=0.5cm]
            
            \foreach \x in {0,0.5}{
                        \foreach \y in {0,0.5}
                            \node[box, fill=cyan] at (\x,\y){};
                    }
            \foreach \x in {1,1.5}{
                        \foreach \y in {0,0.5}
                            \node[box, fill=red] at (\x,\y){};
                    }
            \foreach \x in {0,0.5}{
                        \foreach \y in {-1, -0.5}
                            \node[box, fill=yellow] at (\x,\y){};
                    }
            \foreach \x in {1,1.5}{
                        \foreach \y in {-1,-0.5}
                            \node[box, fill=orange] at (\x,\y){};
                    }

            \node[font=\fontfamily{phv}\selectfont] at (0.75, -1.5) {Input};
            \draw[->] (2.25, -0.25) to (2.75,-0.25);

            \node at (3.5, 0.5)  (M1) {$(U_{1,0}$};
            \node at (5, 0.5)  (M1) {$(U_{0,0} R_y(\tikz{\node[rectangle,draw=black,thick,minimum size=0.25cm, fill=cyan, inner sep=0pt, yshift=-4ex] {};}))$};
            \node at (7, 0.5)  (M2) {$(U_{0,1} R_y(\tikz{\node[rectangle,draw=black,thick,minimum size=0.25cm, fill=cyan, inner sep=0pt, yshift=-4ex] {};}))$};
            \node at (9, 0.5)  (M3) {$(U_{0,2} R_y(\tikz{\node[rectangle,draw=black,thick,minimum size=0.25cm, fill=cyan, inner sep=0pt, yshift=-4ex] {};}))$};
            \node at (11.1, 0.5)  (M4) {$(U_{0,3} R_y(\tikz{\node[rectangle,draw=black,thick,minimum size=0.25cm, fill=cyan, inner sep=0pt, yshift=-4ex] {};})))\cdot$};

            \node at (3.5, 0)  (M1) {$(U_{1,1}$};
            \node at (5, 0)  (M1) {$(U_{0,0} R_y(\tikz{\node[rectangle,draw=black,thick,minimum size=0.25cm, fill=red, inner sep=0pt, yshift=-4ex] {};}))$};
            \node at (7, 0)  (M2) {$(U_{0,1} R_y(\tikz{\node[rectangle,draw=black,thick,minimum size=0.25cm, fill=red, inner sep=0pt, yshift=-4ex] {};}))$};
            \node at (9, 0)  (M3) {$(U_{0,2} R_y(\tikz{\node[rectangle,draw=black,thick,minimum size=0.25cm, fill=red, inner sep=0pt, yshift=-4ex] {};}))$};
            \node at (11.1, 0)  (M4) {$(U_{0,3} R_y(\tikz{\node[rectangle,draw=black,thick,minimum size=0.25cm, fill=red, inner sep=0pt, yshift=-4ex] {};})))\cdot$};

            \node at (3.5, -0.5)  (M1) {$(U_{1,2}$};
            \node at (5, -0.5)  (M1) {$(U_{0,0} R_y(\tikz{\node[rectangle,draw=black,thick,minimum size=0.25cm, fill=yellow, inner sep=0pt, yshift=-4ex] {};}))$};
            \node at (7, -0.5)  (M2) {$(U_{0,1} R_y(\tikz{\node[rectangle,draw=black,thick,minimum size=0.25cm, fill=yellow, inner sep=0pt, yshift=-4ex] {};}))$};
            \node at (9, -0.5)  (M3) {$(U_{0,2} R_y(\tikz{\node[rectangle,draw=black,thick,minimum size=0.25cm, fill=yellow, inner sep=0pt, yshift=-4ex] {};}))$};
            \node at (11.1, -0.5)  (M4) {$(U_{0,3} R_y(\tikz{\node[rectangle,draw=black,thick,minimum size=0.25cm, fill=yellow, inner sep=0pt, yshift=-4ex] {};})))\cdot$};

            \node at (3.5, -1)  (M1) {$(U_{1,3}$};
            \node at (5, -1)  (M1) {$(U_{0,0} R_y(\tikz{\node[rectangle,draw=black,thick,minimum size=0.25cm, fill=orange, inner sep=0pt, yshift=-4ex] {};}))$};
            \node at (7, -1)  (M2) {$(U_{0,1} R_y(\tikz{\node[rectangle,draw=black,thick,minimum size=0.25cm, fill=orange, inner sep=0pt, yshift=-4ex] {};}))$};
            \node at (9, -1)  (M3) {$(U_{0,2} R_y(\tikz{\node[rectangle,draw=black,thick,minimum size=0.25cm, fill=orange, inner sep=0pt, yshift=-4ex] {};}))$};
            \node at (11.05, -1)  (M4) {$(U_{0,3} R_y(\tikz{\node[rectangle,draw=black,thick,minimum size=0.25cm, fill=orange, inner sep=0pt, yshift=-4ex] {};})))$};
        \end{tikzpicture}
        }
    \end{minipage}
    }
    \caption{A two-layer fragment encoding mimicking two $2\times2$ convolutions with a stride of 2. (a) Illustrates the information abstracted from the input layer to the second layer. (b) Shows the fragment encoding utilizing $R_y$ for the encoding and $U_{i,j}$ as the $j$-th weight matrix of the kernel in layer $i$ applied to the input.}
    \label{fig:rqcnn_reupload_pipeline_multi_conv_2L_red}
\end{figure}

In a classical CNN, the weights $w$ of a $k\times k$ convolutional kernel in layer $l$ are multiplied by the corresponding layer inputs $x$ at position $(i,j)$, and the results are convolved by summation:
\begin{equation}
    x^{(l+1)}_{i,j}=\sum_{a=1}^{k}\sum_{b=1}^{k} w_{a,b}x^{(l)}_{i+a,j+b}.
\end{equation}
In the QCNN, we represent both weights and inputs using parameterized single-qubit gates. Since each gate represents a matrix, the associative law allows us to perform weight multiplication through a sequence of single-qubit gates:
\begin{multline}
    U(\phi_0)E(x_0)\dots U(\phi_{k^2})E(x_{k^2})=\\
    (U(\phi_0)E(x_0))\dots (U(\phi_{k^2})E(x_{k^2})),
\end{multline}
where $U(\phi_i)$ represents the weights and $E(x_i)$ the inputs for a $k\times k$ kernel. The results of these weight multiplications are convolved by multiplying the matrices together. More formally, the encoded inputs $E(x^{(l)}_{i+a,j+b})$ are processed by applying parameterized unitaries $U(\phi^{(l)}_{a,b})$ within a $k\times k$ kernel window such that
\begin{equation}
    E(x_{i,j}^{(l+1)})=\prod_{a=1}^{k}\prod_{b=1}^{k} U(\phi_{a,b}^{(l)})E(x_{i+a,j+b}^{(l)}).
\end{equation}
An example of this process is shown in \Cref{fig:fragment_enc}. 

For the initial encoding $E(x_{i,j}^{(0)})$, single-qubit encoding methods like QE, dense qubit encoding (DQE), weighted universal encoding (WUE) and universal encoding (UE) can be used (see \Cref{app:encoding} for details). The only constraint we impose is that each input gate is followed by a trainable gate. For instance, in DQE, two input values are encoded at once. To simulate a $2 \times 2$ convolution, two parameterized quantum gates are used to process the four input elements (\Cref{app:qcnn_setup}, \Cref{fig:rqcnn_reupload_pipeline_conv_1L_rzry}). 

The encoding can be scaled through a trade-off between the number of qubits and the number of layers. More layers reduce the number of required qubits. An important aspect to note is that independent of the number of layers, Proposition \ref{prp:manyGates2one} ensures that the physical implementation requires only a single universal gate.

\begin{proposition}   
\label{prp:manyGates2one}
    Every sequence of single-qubit gates \( U_1, U_2, \dots, U_n \) can be combined into a single equivalent universal $U3$ gate.
\end{proposition}
\begin{proof}
            For a sequence  $U_1, U_2, \dots, U_n \in$ SU(2) of single-qubit gates, their product remains in SU(2), as SU(2) is closed under multiplication. Since $U3$ can represent any single-qubit operation in SU(2):
            \begin{equation}
            \exists \theta,\phi,\lambda\in\mathbb{R} : U3(\theta,\phi,\lambda)=U_n U_{n-1} \cdots U_0.
            \end{equation}
\end{proof}

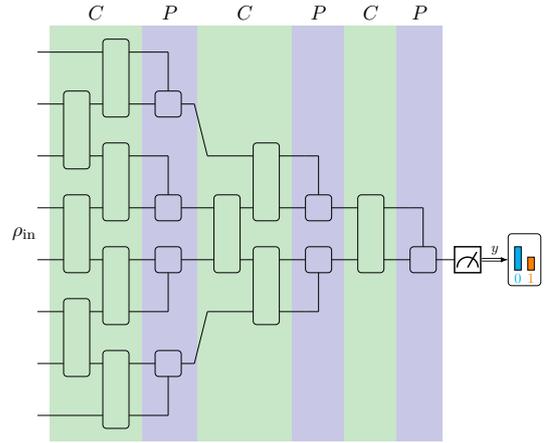
\begin{figure}[t]
    \centering
    \resizebox{0.425\textwidth}{!}{
    \begin{tikzpicture}
    \definecolor{ccolor}{RGB}{200,230,200}
    \definecolor{pcolor}{RGB}{200,200,230}
    \definecolor{fcolor}{RGB}{255,230,210}
    
    \fill[ccolor] (1.25,0) rectangle (3,-8);
    \fill[pcolor] (3,0) rectangle (4.0625,-8);
    \fill[ccolor] (4.0625,0) rectangle (5.875,-8);
    \fill[pcolor] (5.875,0) rectangle (6.875,-8);
    \fill[ccolor] (6.875,0) rectangle (7.875,-8);
    \fill[pcolor] (7.875,0) rectangle (8.75,-8);
    
    \node[above] at (2.125,0) {$C$};
    \node[above] at (3.53125,0) {$P$};
    \node[above] at (4.96875,0) {$C$};
    \node[above] at (6.375,0) {$P$};
    \node[above] at (7.375,0) {$C$};
    \node[above] at (8.3125,0) {$P$};

    \node[above] at (0.75,-4.25) {$\rho_{\text{in}}$};

    \foreach \y in {-0.5,-1.5,-2.5,-3.5,-4.5,-5.5,-6.5,-7.5} {
        \draw (1,\y) -- (1.5,\y);
    }
    
    \foreach \y in {-2.25,-4.25,-6.25} {
        \draw[rounded corners=2pt] (1.5,\y-0.5) rectangle (2,\y+1);
    }
    
    \foreach \y in {-1.25,-3.25,-5.25,-7.25} {
        \draw[rounded corners=2pt] (2.25,\y-0.5) rectangle (2.75,\y+1);
    }
    
    \draw (1.5,-0.5) -- (2,-0.5);
    \draw (1.5,-7.5) -- (2,-7.5);
    
    \foreach \y in {-0.5,-1.5,-2.5,-3.5,-4.5,-5.5,-6.5,-7.5} {
        \draw (2,\y) -- (2.25,\y);
    }
    
    \foreach \y in {-0.5,-1.5,-2.5,-3.5,-4.5,-5.5,-6.5,-7.5} {
        \draw (2.75,\y) -- (3.25,\y);
    }
    
    \foreach \y in {-1.75,-3.75,-4.75,-6.75} {
        \draw[rounded corners=2pt] (3.25,\y) rectangle (3.75,\y+0.5);
    }
    
    \foreach \y in {-0.5,-2.5,-5.5,-7.5} {
        \draw (3.25,\y) -- (3.5,\y);
    }
    
    \foreach \y in {-0.5,-2.5} {
       \draw (3.5,\y) -- (3.5,\y-0.75);
    }
    
    \foreach \y in {-5.5,-7.5} {
       \draw (3.5,\y) -- (3.5,\y+0.75);
    }
    
    \foreach \y in {-1.5,-3.5,-4.5,-6.5} {
        \draw (3.75,\y) -- (4,\y);
    }
    
    \foreach \y in {-3.5,-4.5} {
        \draw (4,\y) -- (4.375 ,\y);
    }
    
    \draw (4,-1.5) -- (4.25,-2.5);
    \draw (4,-6.5) -- (4.25,-5.5);
    
    \foreach \y in {-2.5,-5.5} {
        \draw (4.25,\y) -- (5.125,\y);
    }
    
    \foreach \y in {-3.5,-4.5} {
        \draw (4.875,\y) -- (5.125,\y);
    }
    
    \draw[rounded corners=2pt] (4.375,-4.75) rectangle (4.875,-3.25);
    \draw[rounded corners=2pt] (5.125,-3.75) rectangle (5.625,-2.25);
    \draw[rounded corners=2pt] (5.125,-5.75) rectangle (5.625,-4.25);
    
    \foreach \y in {-2.5,-3.5,-4.5,-5.5} {
        \draw (5.625,\y) -- (6.125,\y);
    }
    
    \draw[rounded corners=2pt] (6.125,-3.75) rectangle (6.625,-3.25);
    \draw[rounded corners=2pt] (6.125,-4.75) rectangle (6.625,-4.25);
    
    \foreach \y in {-2.5,-5.5} {
        \draw (6.125,\y) -- (6.375,\y);
    }
    
    \draw (6.375,-2.5) -- (6.375,-3.25);
    \draw (6.375,-5.5) -- (6.375,-4.75);
    
    \foreach \y in {-3.5,-4.5} {
        \draw (6.625,\y) -- (7.125,\y);
    }
    \draw[rounded corners=2pt] (7.125,-4.75) rectangle (7.625,-3.25);
    
    \foreach \y in {-3.5,-4.5} {
        \draw (7.625,\y) -- (8.125,\y);
    }
    \draw (8.125,-3.5) -- (8.375,-3.5);
    \draw (8.375,-3.5) -- (8.375,-4.25);
    \draw[rounded corners=2pt] (8.125,-4.75) rectangle (8.625,-4.25);
    
    \draw (8.625,-4.5) -- (8.975,-4.5);
    \draw[thick] (8.975,-4.75) rectangle (8.975+0.5,-4.25);
    \draw[thick] (9.025,-4.65) arc[start angle=180,end angle=0,radius=0.2cm];
    \draw[thick] (9.225,-4.65) -- (9.4,-4.35);
    
    \draw[rounded corners=2pt] (10,-5) rectangle (10.625,-4);
    \draw[fill=cyan] (10.125,-4.7) rectangle (10.25,-4.25) node[text=cyan, yshift=-3pt, scale=0.7] at (10.1875,-4.75) {0};
    \draw[fill=orange] (10.375,-4.7) rectangle (10.5,-4.45) node[text=orange, yshift=-3pt, scale=0.7] at (10.4375,-4.75) {1};
    
    \draw [line width=0.5pt, double distance=1pt,
                 arrows = {-Latex[length=0pt 2 0]}] 
                 (9.5,-4.5) -- (10,-4.5) node[above, midway, scale=0.7] {$y$};
    \end{tikzpicture}
    }
    \caption{Quantum convolutional neural network architecture. The information is processed through alternating layers of quantum convolutional gates (green squares) and quantum pooling operations (blue squares). The final measurement operation (right) collapses the quantum state and outputs a classical prediction $y$, which is used for classification.}
    \label{fig:rqcnn_qcnn}
\end{figure}

After the encoding, the QCNN is applied (\Cref{fig:rqcnn_qcnn}). To identify suitable convolution circuits for the QCNN, the expressibility is defined using the discretized fidelity distribution between pairs of quantum states drawn from the PQC and the Haar-random states \cite{sim2019expressibility}
\begin{equation}
    \text{expr}=\frac{1}{|C|}\sum_{\ket{\psi}\in C}D_{KL}(\hat{P}_{PQC}(F;S,\ket{\psi})||P_{Haar}(F)),
\end{equation}
where input states $C$ are drawn from the Haar measure. Moreover, the interval of the discretized fidelity distribution is fixed to a specific range such that every bin has a probability of at least \(\epsilon_{bin}\). The range of considered fidelities is truncated to satisfy \(\epsilon_{bin}\) and renormalized to maintain a valid probability distribution. If a PQC expresses fidelities beyond this truncated interval, the circuit is denoted with an infinite expressibility value. This process prevents trivial PQCs from becoming expressible simply by increasing the number of qubits. For instance, sampled states from a 15 qubit PQC with $R_z$ gates and initial states $\ket{+}$ would appear to have a fidelity distribution of Haar-random states using 75 histogram bins ($D_{KL}(\hat{P}_{PQC}(F;S,\ket{+})||P_{Haar}(F))\approx 0$) despite exploring only the equator of each Blochsphere.

 
\subsection{Regular QCNN results} 
\begin{table*}[!t]
    \centering
    \begin{tabular}{c c C}
        \toprule
        Number of Qubits & Best Model & \text{Accuracy} \\ 
        \midrule
        1 & $R_x-R_y-R_z-R_x-R_y \rightarrow U3$ & 51.35 \\ 
        4 & $U3-U3-U3-U3 \rightarrow C5$ & 61.37 \\ 
        16 & $U3-U3-U3 \rightarrow Pool-C2$ & 71.79 \\ 
        \bottomrule
    \end{tabular}
    \caption{Best models from grid search results of the regular QCNN. In the Best Model column, the first sequence represents the fragment encoding, followed by '$\rightarrow$' indicating transition to the QCNN. Models are read left to right, showing layer order.}
    \label{tab:grid_search}
\end{table*}

Based on our objective of expressibility, entanglement and complexity, our ansatz search discovered circuits that consistently achieved better overall performance compared to circuits from previous research, similar to results previously demonstrated for the hybrid QCNN (\Cref{fig:ansatz_search_res}). 



Considering both QCNN architectures, the efficacy of circuits from previous research varied significantly depending on implementation in hybrid versus regular QCNN. For example, circuit 1 (C1) exhibited strong performance in regular QCNN implementations but performed poorly in hybrid architectures. These findings indicate that circuits with promising results in previous research may yield suboptimal performance when applied to different problem settings, supporting the 'no free lunch' theorem from classical machine learning. To design such ansätze, our experiments identified crucial elements for effective circuit design, such as parameter sharing (\Cref{fig:convolutional_circuits_2q}) and biases of initializations that are provided with all results of the ansatz search in Section SIV of the Supplemental Material.

Utilizing the quantum circuits from the ansatz search and previous literature for the convolution layers (\Cref{fig:convolutional_circuits_2q}), we constructed a grid search to identify the required number of qubits for the QCNN. The analysis explored 1-, 4-, and 16-qubit models for classification using padded MNIST images of $32\times32$ pixels, corresponding to compressing 1024, 256, and 64 pixels into a single qubit, respectively. Each layer in the fragment encoding functions similarly to a $2\times2$ convolution with a stride of 2, combining information from 4 units of the previous layer into one unit in the current layer. QE was employed for encoding the input data. For the trainable single-qubit gates applied in each layer, three different architectures were tested: (1) alternating rotation gates, (2) alternating rotation and universal gates, or (3) universal gates (see \Cref{app:qcnn_setup}, \Cref{tab:qubits_pipeline_layers_setup}).

After the fragment encoding, the QCNN layers are attached. For the 1-qubit case, a $1\times1$ convolution layer of a single-qubit universal gate was applied. For the 4-qubit case, both a $2\times2$ convolution layer and a $2\times2$ pooling layer with a $1\times1$ convolution layer were tested. The 16-qubit architectures included three variants: $2\times2$ pooling followed by a $2\times2$ convolution layer, two successive $2\times2$ pooling operations with a $1\times1$ convolution layer, and an interpolation followed by a $3\times3$ convolution layer. The interpolation was implemented through a one-dimensional pooling circuit given by \Cref{eq:pooling}, where the last row and column are pooled into the adjacent inner row and column (\Cref{app:qcnn_setup}, \Cref{fig:interpolation_scheme}). All circuits from \Cref{fig:convolutional_circuits_2q} were tested for every $2\times 2$ (4 qubits) and $3\times 3$ (9 qubits) convolution and a single-qubit universal gate for the $1\times 1$ convolution (see \Cref{app:qcnn_setup}, \Cref{tab:qubits_pipeline_layers_setup}). Each architecture was run once.


In \Cref{tab:grid_search}, the best model for each respective number of qubits is presented. For the 1-qubit model, the accuracy appears to be close to random. However, for the 4- and 16-qubit models, the accuracy increases by approximately $10\%$ with each increment in the number of qubits. While none of these results outperformed the classical CNN, the observed performance improvements with increased qubit count motivated exploration beyond the original search space. Although 64-qubit models were computationally infeasible, the matrix product state (MPS) \cite{vidal2003efficient} simulation method enabled analysis of 49-qubit models that exhibit only slight entanglement, compressing 16 inputs per qubit for $28\times28$ images.



\begin{figure*}[!t]
  \centering
  \subfloat{
    \centering
    \begin{tikzpicture}
      \node[inner sep=0pt, anchor=south west] (img) 
        {\includegraphics[width=0.48\linewidth]{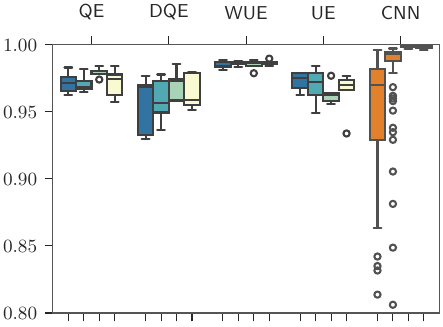}};
    \end{tikzpicture}
  }%
  \hfill
  \subfloat{
    \centering 
    \begin{tikzpicture}
      \node[inner sep=0pt, anchor=south west] (img) 
        {\includegraphics[width=0.48\linewidth]{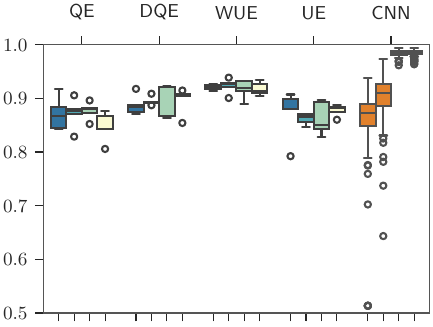}};
    \end{tikzpicture}
  }

  \vskip\baselineskip
  \subfloat{
    \centering 
    \begin{tikzpicture}
      \node[inner sep=0pt, anchor=south west] (img) 
        {\includegraphics[width=0.48\linewidth]{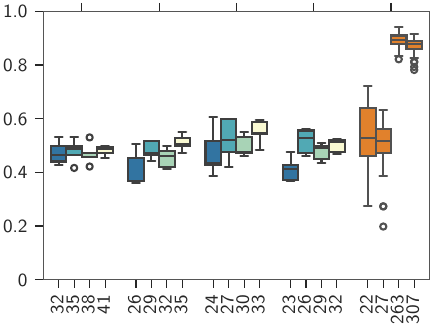}};
    \end{tikzpicture}
  }%
  \hfill
  \subfloat{
    \centering 
    \begin{tikzpicture}
      \node[inner sep=0pt, anchor=south west] (img) 
        {\includegraphics[width=0.48\linewidth]{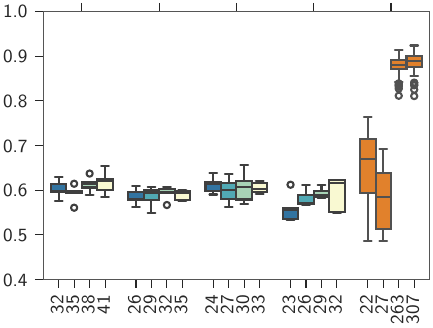}};
    \end{tikzpicture}
  }

  \begin{tikzpicture}[remember picture, overlay]
    \node at (-8.5,13.5) {(a)};
    \node at (8.5,13.5) {(b)};
    \node at (-8.5,0.5) {(c)};
    \node at (8.5,0.5) {(d)};
    
    \node[font=\fontfamily{phv}\selectfont, rotate=90] at (-9,7) {Accuracy};
    \node[font=\fontfamily{phv}\selectfont] at (0,0) {Parameters};
    \node[font=\fontfamily{phv}\selectfont] at (0,14) {Embedding};
    \node at (9.25,7) {\includegraphics[width=.1\textwidth]{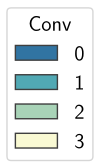}}; 
    
  \end{tikzpicture}

  \caption{Performance comparison of different quantum embeddings and $1\times1$ convolutions with $U3$ gates across four classification tasks. Subfigure (a) shows the results for 0 vs 1, (b) for 7 vs 8, (c) for 0–1–2–3, and (d) for greater than 4 (g4), where the numbers 22, 27, 263, and 307 indicate the parameter counts of the CNN baselines tailored for each task. $U3$ gates are applied from the last layer to the first.}
  \label{fig:rqcnn_q49_embedding}
\end{figure*}

Due to the low entanglement constraint of MPS, the experiments focused on interpolations with pooling layers and alternating $1\times1$ convolution layers. To explore a broader variety of regular QCNN models, variations of different feature embeddings for the fragment encoding were explored: QE and DQE with 2 layers of $2\times2$ convolutions, and UE and WUE with $3\times3$ followed by $1\times2$ convolution layer. Based on the grid search results showing superior performance with universal gates, $U3$ gates were used in the fragment encoding. 

The regular QCNN models surpassed the CNN for the binary digit classification of 0 and 1 (\Cref{fig:rqcnn_q49_embedding}\textcolor{red}{a}), with the most effective implementation utilizing WUE with three additional \( U3 \) layers. This model achieved $98.7\%$ accuracy with merely $0.2\%$ standard deviation, demonstrating superior reliability and consistency compared to the CNN, which attained $93.4\%$ test accuracy with a $10.6\%$ standard deviation. While the inclusion of $1\times1$ convolutions between pooling layers showed no significant performance improvement, all QCNN variants demonstrated accelerated convergence compared to the classical CNN, aligning with current literature \cite{abbas2021power, caro2022generalization}. For instance, \Cref{fig:rqcnn_embedding_01_comparison} illustrates the convergence speed of the model with WUE embedding and no $1\times1$ convolutions. 

Performance analysis of regular QCNN models across increasingly complex classification tasks revealed varying effectiveness compared to CNN baselines (\Cref{fig:rqcnn_q49_embedding}). Overall the WUE embedding performed the best. For classifying digits 7 and 8, the QCNN models marginally outperformed the CNN, achieving $92.1\%$ accuracy compared to the CNN’s $90.0\%$ accuracy (27 parameters). However, for more complex tasks, the QCNNs showed significantly lower performance. In digits 0-3 classification, they reached only $55.1\%$ accuracy compared to CNN's $89.2\%$ (263 parameters), and in greater-than-4 classification, they peaked at $63.2\%$ versus CNN's $88.5\%$ (307 parameters). 

The significant disparity in parameter count between QCNN and CNN architectures suggests that more complex QCNN structures may be necessary for fair comparison. Despite current limitations in complex classification tasks, the demonstrated advantages in convergence speed and binary classification performance for digits 0-1 and 7-8 indicate potential for improvement on classical state-of-the-art models. 

\begin{figure}[t]
    {
    \captionsetup{oneside,margin={1cm,0cm}}
    \subfloat[Test accuracy\label{fig:sub3}]{
        \includegraphics[width=0.4285\textwidth]{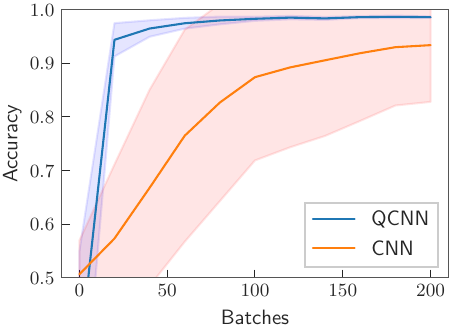} 
    }
    \hfill
    
        \hspace*{-0.65em}
    \subfloat[Test loss\label{fig:sub4}]{
        \centering
        \includegraphics[width=0.4285\textwidth]{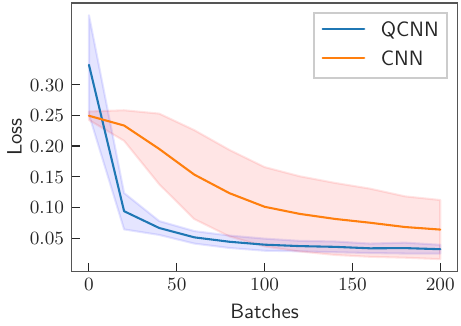} 
    }
    }
    \caption{Comparison of convergence behavior between regular QCNN models without $1\times1$ convolutions and WUE embedding versus the classical CNN for classifying 0 and 1. The shaded regions indicate ±1 standard deviation from the mean across experimental runs.
    }
    \label{fig:rqcnn_embedding_01_comparison}
\end{figure}

\section{Quantum computer experiment} 
To verify the results, the regular QCNN was evaluated on IBM's Heron r2 \cite{IBMQuantum} quantum processor (156 qubits, July 2024) for the simplest case of 0 vs 1 classification. To ensure a fair or CNN-favorable comparison, the parameters of the QCNN were reduced to 20 by using an $R_x$ rotation gate in the second layer of the WUE fragment encoding and omitting the $1\times1$ convolution (\Cref{app:qc_setup}, \Cref{tab:rqcnn_embedding}). Moreover, to reduce the number of additional gates needed for transpilation to the quantum computer, the $CRZ$ gate of the pooling layer was replaced by the native gates $CZ$ and $R_z$ of the system. These modifications did not result in a loss of accuracy or convergence speed (\Cref{app:qc_setup}, \Cref{fig:rqcnn_wue_rx_cz_rz}). In fact, later experiments showed that removing the second layer of the fragment encoding entirely did not affect the QCNN's accuracy with the WUE embedding (\Cref{app:qc_setup}, \Cref{tab:rqcnn_q49_embedding_01_compact}). Due to cost efficiency, only the first 50 batches were computed, each using the default setting of 4,096 shots. The test accuracy and test loss were recorded after each batch. Additionally, a simulation of the regular QCNN was run with the current parameters of each batch to monitor the deviation from the actual result. Further details of the quantum computing setup are provided in \Cref{app:qc_setup}.

The regular QCNN demonstrated superior performance with $96.08\%$ accuracy compared to the optimized CNN's $71.74\%$ accuracy (\Cref{fig:qcnn_computation}). The quantum model exhibited rapid convergence after 20 batches and achieved a final test loss of $0.079$ (versus CNN's $0.175$), indicating more effective learning despite operating with compressed input (49 qubits for 784 pixels). The initial lag in accuracy improvement during the first 20 batches, despite the decreasing loss, reflects the binary nature of accuracy metrics, where all predictions exceeding the classification threshold are weighted equally (e.g., $0.51$ and $0.99$).

During training, the outputs from the quantum computer deviated from the simulated predictions by at most $0.05$ for a single batch, with an average deviation of $0.037$ ($\sigma=0.026$) across all batches. The model required no additional error mitigation techniques. This implementation is particularly significant as it represents an effective quantum mechanical encoding that eliminates the need for classical preprocessing and the regular QCNN is able to outperform the CNN despite extensive optimization. This includes an architecture search of over 50,000 configurations, selecting an activation function (ELU) specifically suited for this task, utilizing an optimizer (ADAM) known for achieving highly optimal CNNs, benefiting from extensive literature and advanced frameworks for building effective classical CNN architectures and processing the complete input size. In contrast, the regular QCNN was minimally optimized through different embeddings and restricted to interpolation and pooling operations without convolutions larger than $1\times1$, while processing a reduced input size of 49 qubits. Nonetheless,  \Cref{tab:grid_search} clearly shows improvement with increasing number of qubits.

\begin{figure}[t]
    {
    \captionsetup{oneside,margin={1cm,0cm}}
    \subfloat[Test accuracy]{
        \centering
        \includegraphics[width=0.4285\textwidth]{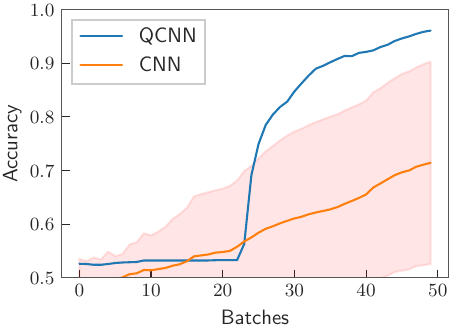}
    }
        \hfill
        \hspace*{-0.65em}
    \subfloat[Test loss]{
            \includegraphics[width=0.4285\textwidth]{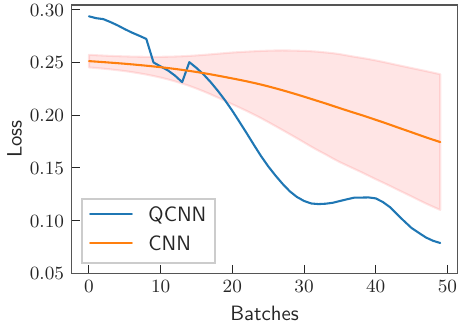}
    }
    }
    \caption{
        Experiment on the Heron r2 quantum computer comparing the developed 49-qubit regular QCNN model with the optimized CNN baseline. Subfigure (a) presents the test accuracy, while subfigure (b) displays the test loss. The shaded region indicates ±1 standard deviation from the mean across experimental runs.}
        \label{fig:qcnn_computation}
\end{figure}

\section{Discussion}
Our experimental results demonstrate that QCNNs can outperform classical CNNs in binary classification tasks, even when operating under significant hardware constraints and minimal architecture optimization. This achievement is particularly significant as it suggests quantum advantages may be attainable in practical machine learning applications, even on current NISQ devices with limited qubit counts.

The fully quantum mechanical approach (regular QCNN) proved more effective than the hybrid QCNN. The fragment encoding enabled quantum encoding without classical preprocessing \textendash{} a significant advancement for NISQ implementation. This success suggests that a sequence of trainable single-qubit gates, like in fragment encoding or data re-uploading \cite{perez2020data} could be an efficient encoding strategy for variational quantum algorithms in general. Remarkably, the advantages of the regular QCNN shown in simulation compared to classical CNNs were confirmed on actual quantum hardware despite noise, showing minimal deviation. This indicates that QCNNs already offer substantial advantages for small problem sets over classical architectures, with the potential for even greater improvements as quantum hardware advances to support larger convolution operations and more qubits, particularly given the clear pattern of improvement observed with increased qubit count.

The ansatz search methodology, while secondary, demonstrated effectiveness in balancing circuit complexity with performance metrics for both architectures. Parameter sharing emerged as a crucial feature in PQC design, and results showed that circuits do not perform equally well for all architectures, similar to the 'no free lunch' theorem in classical machine learning. This methodology could be extended beyond QCNNs to design quantum circuits for other quantum algorithms, such as the variational quantum eigensolver, which is widely used to compute the ground states of molecular systems \cite{tilly2022variational}. 

While classical hardware continues to evolve, there will be domains where classical computers struggle, particularly in high-dimensional spaces where (hybrid) QCNNs can operate effectively. The faster convergence observed in QCNNs suggests potential benefits for quantum-enhanced training of larger models, possibly reducing the intensive computational resources currently required for tasks like training large language models.

\section*{Acknowledgment}
We acknowledge the use of IBM Quantum services for this work. The views expressed are those of the authors, and do not reflect the official policy or position of IBM or the IBM Quantum team.

\bibliographystyle{quantum}
\bibliography{main}

\appendix
\section{Bayesian optimization setup}
\label{app:bo_setup}
This appendix provides the technical details regarding the Bayesian optimization setup employed in our experiments, specifically addressing the CNN baseline and the ansatz search configurations.
\subsection{CNN baseline setup}
\begin{proposition}
\label{prp:bsoc}
        Given a differentiable non-negative loss function $L$ for any sample $(x,y_i)$:
        \begin{itemize}
            \item Only for a false prediction with maximum certainty ($f_\theta(x)=c_j, i\neq j$): 
            \begin{equation}
                L(y_i, f_\theta(x)) = \max_{\hat{y}\in[0,1]}L(y_i, \hat{y}).
            \end{equation}
            \item Only for a true prediction with maximum certainty ($f_\theta(x)=c_i$): 
            \begin{equation}
                L(y_i, f_\theta(x)) = 0.
            \end{equation}
        \end{itemize}
        leads either to vanishing gradients or local minima for the BSOC problem.
        \label{prop:bin_clf}
    \end{proposition}
    \begin{proof}
        Consider target label $y_i\in \mathcal{Y}$ and $S\subseteq B_j$ with $i<j$ w.l.o.g., where $c_k$ are the corresponding center and $\forall s\in S: c_j<s$. There are two cases:
        \\
        Case 1: The loss is constant for $S$:
            \begin{equation}
                \forall x\in \mathcal{X}: f_\theta(x)\in S \implies \frac{d}{dx}L(y_i,f_\theta(x))=0.
            \end{equation}
        \\  
        Case 2: The loss varies within $S$:
            \begin{multline}    
                \forall x\in \mathcal{X}: f_\theta(x)\in S \implies 
                (L(y_i,f_\theta(x))<L(y_i,c_j))
                \\
                \wedge (L(y_i,f_\theta(x))<L(y_i,c_{j+1})\ w.l.o.g.).
            \end{multline}
            But the local minima is not global since
            \begin{equation}
                \forall f_\theta(x)>c_j \implies L(y_i,f_\theta(x))>L(y_i,c_i).
            \end{equation}
            
    \end{proof}

The Bayesian optimization used for searching the CNN baseline prunes models that have been proposed more than 10 times. For the training a learning rate of $0.01$ over 200 batches of size 25, each randomly sampled as in Ref. \cite{hur2022quantum} was set. Analog to the estimation of the maximum number of parameters, we set constraints of 20 layer depth, 128 channels, stride-1 convolutions (max. $3\times3$, padding 1), and $2\times2$ max pooling with stride 2, which showed sufficient to classify the entire MNIST dataset using a modified model from Ref. \cite{wang2020regression}.

\subsection{Ansatz search setup}
\label{app:bo_setup_as}
The ansatz search for the hybrid QCNN was tested with qubit sizes of 2 and 3 (Type I), and 4 and 9 (Type II), as detailed in Section SIV of the Supplemental Material. The expressibility threshold was established by applying Gaussian noise ($\sigma_{std}=0.05$) to a uniform target distribution over four classes, corresponding to the maximum number of classes in our classification tasks. This approach yields a worst-case expressibility of approximately 1.39 and maintains compatibility with binary classification tasks, which would not be the case for three classes. Note, since outputs from one layer serve as inputs to the next layer's PQC, and PQCs were found to be expressible given continuous uniform random inputs, using even higher multiples of the class count as target bins might be favorable to ensure uniform distributions within bins and thus maintain expressibility throughout the network.

The regular QCNN implementation utilized qubit sizes of 4 and 9 for standard operations. Additionally, the results of the ansatz search for 2- and 3-qubit configurations are computed as a proof-of-concept (Section SIV of the Supplemental Material). Following Ref. \cite{sim2019expressibility}, the system employed 75 expressibility bins. For 2-4 qubits, the full fidelity distribution was utilized, while 9-qubit probabilities were truncated ($\epsilon_{bin}=10^{-30}$), resulting in worst-case expressibility values of $12.95$, $30.22$, $69.08$, and $69.08$ for 2, 3, 4, and 9 qubits, respectively. The expressibility threshold was set using Gaussian noise ($\sigma_{std}=0.2$).

The entanglement threshold is for both the hybrid QCNN and regular QCNN given by the mean entanglement of a Haar-random pure state:
\begin{equation}
    \text{entgl}_{thr}=\langle Q\rangle_{Haar}=\frac{N-2}{N+1},
\end{equation}
where \(N\) represents the dimension of the Hilbert space. This threshold was also used in Ref. \cite{weinstein2008parameters,sim2019expressibility} for circuit comparison. 

The parameter space of the ansatz scales with qubit size $q$, setting maximum parameters to $q$, circuit depth to $3q$, and gate count to $5q$. The system allowed for 10 maximum duplicate circuits. The sampling process involved 10 random inputs with 10,000 weights each (100,000 total), with hybrid QCNN using random parameters for input embedding and regular QCNN employing random state vectors from the Haar distribution. The results were averaged across the 10 input configurations. The complete search comprised 10,000 trials. 

The final evaluation of all circuits was averaged across 100 input configurations with 10,000 weights each. Single universal U3 gates were employed for $1\times1$ convolutions. Due to combinatorial complexity final circuit selection only includes optimal expressibility, entanglement, and objective function performance, though parameter settings may benefit from further optimization.

\section{Encoding}
\label{app:encoding} 
The encoding of data into a quantum circuit can be seen as a quantum version of a feature embedding \cite{schuld2019quantum}, a concept widely used in machine learning \cite{bengio2014representationlearningreviewnew, Goodfellow-et-al-2016}. A quantum feature embedding can be defined by $U_{\phi}(x):\mathcal{X} \rightarrow \mathcal{H}$ with $U_{\phi}(x)\ket{\psi}=\ket{\phi(x)}$, where $U_{\phi}(x)$ is the state preparation circuit that transforms the state $\ket{\psi}$ into the state $\ket{\phi(x)}$ in the Hilbert space $\mathcal{H}$. For the following encodings, let \( x = (x_1, \ldots, x_N)^T\in\mathbb{R}^N \) be the input vector and $\sigma_y, \sigma_z$ be Pauli matrices.

The qubit encoding (QE) maps every input element onto a different qubit with a constant depth \cite{grant2018hierarchical, larose2020robust, hur2022quantum}. More precisely, every element in the vector $x\in[0,\pi)^N$ is encoded with 
\begin{equation}
    U_{\phi}(x)=\bigotimes_{j=1}^Ne^{-i\frac{x_j}{2}\sigma_y}.
\end{equation}

The dense qubit encoding (DQE) encodes two real-valued parameters onto a single qubit. While QE only explores a rotation around the y-axis, DQE utilizes the entire Bloch sphere of a qubit by employing two rotations around orthogonal axes \cite{larose2020robust, hur2022quantum}. Here, the $z$ and $y$ axes of the Bloch sphere are exploited such that the input vector $x\in[0,\pi)^N$ can be encoded with
\begin{equation}
    U_{\phi}(x)=\bigotimes_{j=1}^{\lceil N/2 \rceil}e^{-i\frac{x_{2j}}{2}\sigma_z}e^{-i\frac{x_{2j-1}}{2}\sigma_y}.
\end{equation}
If the input size $N$ is odd, a zero is appended to the input vector to form the last component. 

The universal encoding (UE) maps $x\in[0,\pi)^N$ to $\lceil N/3 \rceil$ qubits as follows
\begin{equation}
    U_{\phi}(x)=\bigotimes_{j=1}^{\lceil N/3 \rceil}U3(x_{3j-2},x_{3j-1},x_{3j}),
\end{equation}
where $U3(x_1,x_2,x_3)\in$ SU(2) are arbitrary single-qubit rotation gates defined by
\begin{equation}
    U3(\theta, \phi, \lambda) = \begin{pmatrix}
\cos\left(\frac{\theta}{2}\right) & -e^{i\lambda}\sin\left(\frac{\theta}{2}\right) \\
e^{i\phi}\sin\left(\frac{\theta}{2}\right) & e^{i(\phi+\lambda)}\cos\left(\frac{\theta}{2}\right)
\end{pmatrix}    
\end{equation}
that were also utilized for encoding in Ref. \cite{perez2020data}. The input $x\in[0,\pi)^N$ is split into groups of three. If the input size is not a multiple of three, the remaining elements of the last group are filled with zeroes. 

The weighted universal encoding (WUE) combines the input data with trainable weights before encoding similar to Ref. \cite{perez2020data}. Specifically, the input data are combined as $\theta+w\circ x$, where $\circ$ denotes element-wise multiplication, $\theta$ is a processing angle, and $w$ represents the trainable weights. The state preparation circuit is then given by
\begin{equation}
    U_{\phi}(x)=\bigotimes_{j=1}^{\lceil N/3 \rceil}U3(\theta+w_1x_{3j-2},\theta+w_2x_{3j-1},\theta+w_3x_{3j}).
\end{equation}

\section{QCNN setup}
\label{app:qcnn_setup}
This section provides additional architectural information for both the hybrid and regular QCNNs studied in this work. It assembles the pooling circuit (\Cref{fig:pooling_circuit}), the dense qubit encoding integrated in the fragment encoding (\Cref{fig:rqcnn_reupload_pipeline_conv_1L_rzry}), the interpolation mechanism (\Cref{fig:interpolation_scheme}), the metric values of the convolution operations of the hybrid QCNN (\Cref{tab:mqcnn_metrics_recap}), and the architectures used in the grid search from the regular QCNN (\Cref{tab:qubits_pipeline_layers_setup}) of the main text.

\onecolumn

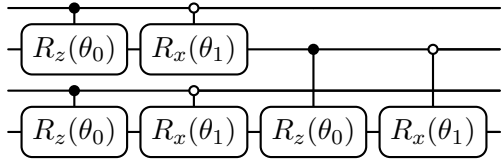
\begin{figure}[ht]
    \centering
    \begin{quantikz}[rounded corners, column sep=5pt, row sep=4pt]
        & \ctrl{1} & \octrl{1} & \qw & \qw & \qw \\
        & \gate{R_z(\theta_0)} & \gate{R_x(\theta_1)} & \ctrl{2} & \octrl{2} & \qw\\
        & \ctrl{1} & \octrl{1} & \qw & \qw & \qw \\
        & \gate{R_z(\theta_0)} & \gate{R_x(\theta_1)} & \gate{R_z(\theta_0)} & \gate{R_x(\theta_1)} & \qw
        \end{quantikz}
    \caption{Quantum pooling circuit. Conditional rotations are applied to the target (lower) qubit based on the control (upper) qubit's state. $R_z$ gate for control $\ket{1}$ (black circle), $R_x$ gate for control $\ket{0}$ (white circle).}
    \label{fig:pooling_circuit}
\end{figure}

\begin{figure}[ht]
    \centering
    \resizebox{0.5\textwidth}{!}{
    \begin{tikzpicture}

    \tikzstyle{box} = [rectangle,draw=black,thick, minimum size=0.5cm]

    \node at (0.25, -1) {Input};

    
    \node[box, fill=yellow] at (0,0) {};
    \node[box, fill=orange] at (0.5,0) {};
    \node[box, fill=cyan] at (0,0.5) {};
    \node[box, fill=red] at (0.5,0.5) {};

    \node at (4, 0.25)  (M1) {$(U_{0} R_z(\tikz{\node[rectangle,draw=black,thick,minimum size=0.25cm, fill=red, inner sep=0pt, yshift=-4ex] {};})R_y(\tikz{\node[rectangle,draw=black,thick,minimum size=0.25cm, fill=cyan, inner sep=0pt, yshift=-4ex] {};}))$};
    \node at (7, 0.25)  (M3) {$(U_{1} R_z(\tikz{\node[rectangle,draw=black,thick,minimum size=0.25cm, fill=orange, inner sep=0pt, yshift=-4ex] {};})R_y(\tikz{\node[rectangle,draw=black,thick,minimum size=0.25cm, fill=yellow, inner sep=0pt, yshift=-4ex] {};}))$};
    \draw[->] (1.25, 0.25) to (1.75,0.25);
    
    \end{tikzpicture}
    }
    \caption{A one-layer fragment encoding with $R_zR_y$, mimicking a $2\times2$ convolution. $U_i$ represents the $i$-th weight matrix of the kernel applied to the input.}
    \label{fig:rqcnn_reupload_pipeline_conv_1L_rzry}
\end{figure}
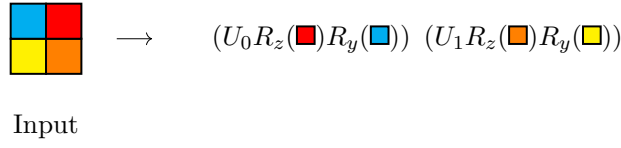

\begin{table}[ht]
\centering
\renewcommand{\arraystretch}{1}
\setlength{\tabcolsep}{8pt}
\begin{tabular}{c C C C C C}
\hline
Model & Circuit ID & Qubits & \text{Expressibility} & \text{Entanglement} & \mathcal{L}_{PQC} \\ 
\hline
\multirow{4}{*}{Exp-Opt} & 5 & 2 & 0.002 \pm 0.001 & 0.313 \pm 0.002 & 1.218 \\ 
& 6 & 3 & 0.0 \pm 0.0 & 0.397 \pm 0.002 & 1.405 \\
& 4 & 4 & 0.002 \pm 0.001 & 0.375 \pm 0.002 & 1.545 \\
& 4 & 9 & 0.002 \pm 0.001 & 0.375 \pm 0.001 & 1.623 \\
\hline
\multirow{4}{*}{Ent-Opt} & 3 & 2 & 0.331 \pm 0.325 & 0.501 \pm 0.14 & 1.23 \\ 
& 5 & 3 & 0.253 \pm 0.006 & 0.626 \pm 0.002 & 1.234 \\
& 5 & 4 & 0.348 \pm 0.005 & 0.711 \pm 0.001 & 1.379 \\
& 3 & 9 & 0.232 \pm 0.299 & 0.776 \pm 0.031 & 1.376 \\
\hline
\multirow{4}{*}{Obj-Opt} & 5 & 2 & 0.002 \pm 0.001 & 0.313 \pm 0.002 & 1.218 \\ 
& 5 & 3 & 0.253 \pm 0.006 & 0.626 \pm 0.002 & 1.234 \\
& 3 & 4 & 0.249 \pm 0.307 & 0.677 \pm 0.062 & 1.348 \\
& 3 & 9 & 0.232 \pm 0.299 & 0.776 \pm 0.031 & 1.376 \\
\hline
\multirow{4}{*}{PQC-Opt} & AS & 2 & 0.013 \pm 0.012 & 0.389 \pm 0.09 & 1.028 \\ 
& AS & 3 & 0.11 \pm 0.223 & 0.758 \pm 0.1 & 1.069 \\
& AS & 4 & 0.006 \pm 0.005 & 0.859 \pm 0.017 & 0.37 \\
& AS & 9 & 0.002 \pm 0.001 & 0.962 \pm 0.002 & 1.032 \\
\hline
\end{tabular}
\caption{Selected quantum circuits in convolution of hybrid QCNNs. The table includes mean and standard deviation of expressibility, entanglement, and objective function values.}
\label{tab:mqcnn_metrics_recap}
\end{table}

\begin{figure}[!ht]
    \centering
    \resizebox{0.125\textwidth}{!}{
        \begin{tikzpicture}
            [
                box/.style={rectangle,draw=black,thick, minimum size=0.5cm},
                filledbox/.style={rectangle,draw=black,thick,fill=brown!50, minimum size=0.5cm},
                arrow/.style={->, thick, shorten >=2pt}
            ]
            
            \foreach \x in {0,0.5,...,1.5}{
                \foreach \y in {0,0.5,...,1.5}{
                    \node[box] at (\x,\y){};
                }
            }
            
            \node[filledbox] at (0,0){};
            \node[filledbox] at (0.5,0){};
            \node[filledbox] at (1,0){};
            \node[filledbox] at (1.5,0){};
        
            \node[filledbox] at (1.5,0.5){};
            \node[filledbox] at (1.5,1){};
            \node[filledbox] at (1.5,1.5){};
            
            \draw[arrow] (0,0.125) -- (0,0.5);  
            \draw[arrow] (0.5,0.125) -- (0.5,0.5);  
            \draw[arrow] (1,0.125) -- (1,0.5);  
        
            \draw[arrow] (1.375,0.5) -- (1,0.5);  
            \draw[arrow] (1.375,1) -- (1,1);  
            \draw[arrow] (1.375,1.5) -- (1,1.5);  
        
            \draw[arrow] (1.375,0) -- (1,0);  
        \end{tikzpicture}
        }
    \caption{The interpolation mechanism. Information from the last row and column is combined into the inner image using a single hierarchy of the pooling circuit from the main text. The arrow in the last corner is applied before the neighboring upward arrow.}
    \label{fig:interpolation_scheme}
\end{figure}
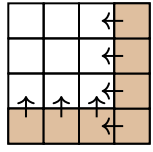

\newpage
\FloatBarrier

\begin{table}[ht]
    \centering
    \begin{tabular}{c c c}
        \toprule
        Number of Qubits & Fragment Encoding & QCNN Layers \\
        \midrule
         1 & \makecell{$R_x-R_y-R_z-R_x-R_y$ \\ $R_x-U3-R_y-U3-R_z$ \\ $U3-U3-U3-U3-U3$} & \( U3 \) \\
        \midrule
         4 & \makecell{$R_x-R_y-R_z-R_x$ \\ $R_x-U3-R_y-U3$ \\ $U3-U3-U3-U3$} & \makecell{ [C$_1$, ..., C$_5$, AS] \\ Pool-\( U3 \) } \\
        \midrule
         16 & \makecell{$R_x-R_y-R_z$ \\ $R_x-U3-R_y$ \\ $U3-U3-U3$} & \makecell{Pool-[C$_1$, ..., C$_5$, AS] \\ Interpol-[C$_1$, ..., C$_5$, AS] \\ Pool-\( U3 \)-Pool-\( U3 \)} \\
        \bottomrule
    \end{tabular}
    \caption{Models used in the grid search, categorized by the number of qubits. For each qubit count, a model consists of one fragment encoding combined with one QCNN layer configuration. C$_i$ refers to the i-th circuit from \Cref{fig:convolutional_circuits_2q}, and '[ ]' denotes a choice of these circuits applied. Pipelines are read from left to right. For example, in the 16-qubit case, one possible model uses the $R_x-U3-R_y$ pipeline with Pool-C$_1$ as its layer configuration.
    }
    \label{tab:qubits_pipeline_layers_setup}
\end{table}

\twocolumn

\section{Classification setup}
\label{app:clf_setup}
The training parameters were set analogously to those used during the Bayesian optimization of the CNN baseline. The test accuracy is recorded every 20th batch. For the CNN architecture, the accuracy was averaged over 100 runs to ensure a robust representation of the training process and results, since some CNN architectures appeared to be volatile during training. The QCNN accuracies were averaged over 5 runs as in Ref. \cite{grant2018hierarchical, hur2022quantum}. In preliminary simulation tests, the regular QCNN model (which was later implemented on the quantum computer) showed consistent performance - when tested with up to fifty runs, its mean accuracy varied by at most one percentage point. All quantum experiments in this study were carried out using Qiskit \cite{qiskit2024}, an open-source framework for quantum information science.

\section{Quantum computing setup}
\label{app:qc_setup}
During initial test runs, uniform noise from $[-0.1, 0.1]$ was applied to the model, which significantly worsened the predictions. This level of noise was consistent with what was observed on the real quantum hardware using a single model. The model consisted of 49 qubits, which allowed it to be mapped up to three times on the backend. However, quantum computers have qubits of varying quality (error rate), and when multiple circuits were run, the results sometimes showed much higher deviations from the expected outcome. This is likely due to the fact that heuristics solving the subgraph isomorphism problem tend to yield better outcomes when more qubits can be ignored or utilized. Additionally, increasing the number of shots to 10,000 was tested, but it did not result in any significant reduction of noise. Therefore, the numerical calculation of the gradient for the ADAM optimizer in Qiskit was adjusted.

In Qiskit, ADAM uses a finite difference of $\epsilon=10^{-10}$ to approximate the gradient. However, if this small change has no significant influence on the result, the noise can dominate and make the gradient approximation unusable. To mitigate this, various finite differences were tested, with $\epsilon \in\{10^{-1}, 10^{-3}, 10^{-5}, 10^{-7}, 10^{-9}\}$. It was found that the noise resistance increases from $\epsilon = 10^{-3}$ onwards, where $\epsilon = 10^{-1}$ was ultimately chosen for its robustness to noise. For the ADAM optimization, the classical CNN relied on analytical gradient computation. An attempt was made to see if the found CNN baseline model could be improved using the numerical gradient computation with $\epsilon = 0.1$, analog to the approach used in the quantum computing case. However, this led to worse outcomes and was therefore disregarded.

Additionally, the neural architecture search with Bayesian optimization for the CNN was conducted twice with a numerical gradient computation and a more focused search space around models with fewer parameters. The first search had an upper bound of 1,000 parameters and a maximum of 16 channels. The second search utilized the original, broader settings of 100,000 parameters as an upper bound, but assigned an accuracy of greater than $90\%$ to any model with more than 10,000 parameters. This assumption was made to speed up the search process for larger models. However, these approaches did not result in a model that was an improvement over the previous CNN baseline.

Other optimizers that are gradient-free, such as COBYLA \cite{Gmez1994AdvancesIO}, appear to be more noise resistant and also show fast convergence for the regular QCNN. However, ADAM was chosen as the optimizer because it is commonly used for classical CNNs and helps maintain the narrative of comparing a near-optimal classical model with the QCNN architecture. The initial parameters were selected from one of 20 previously conducted experimental runs with random parameter initialization, choosing parameters that exhibited a neutral starting condition with test accuracy and test loss close to those of a random classifier. Due to cost constraints, no additional error mitigation techniques were applied. During the experiment the error per layered gate (EPLG) \cite{mckay2023benchmarking}, measured for a 100-qubit chain, was between $0.5$ and $0.6$.

\onecolumn

\begin{table}[t] 
    \centering
    \begin{tabular}{c C C C C}
        \toprule
         Model & \text{Parameters} & \text{Accuracy} & \text{Loss}\\
        \midrule
        $WUE-R_x\rightarrow Pool^*-Pool^*$ & 20  & 98.44 \pm 0.37 & 0.033 \pm 0.008\\ 
        $WUE-R_y\rightarrow Pool^*-Pool^*$ & 20  & 98.37 \pm 0.42 & 0.035 \pm 0.005\\ 
        $WUE-R_z\rightarrow Pool^*-Pool^*$ & 20  & 98.35 \pm 0.34 & 0.041 \pm 0.009\\ 
        CNN  & 22 & 93.36\pm10.57 & 0.064 \pm 0.048\\
        \bottomrule
    \end{tabular}
    \caption{The 49-qubit regular QCNN models with weighted universal embedding. In the Model column, the first sequence represents the fragment encoding, and the '$\rightarrow$' symbol indicates the transition to the QCNN applied afterward. The model is read from left to right, indicating the order of layers. The $Pool^*$ refers to an interpolation followed by a pooling operation.}
    \label{tab:rqcnn_embedding}
\end{table}

\twocolumn

\begin{figure*}[t]
    {
    \captionsetup{oneside,margin={1cm,0cm}}
    \subfloat[Test accuracy]{
        \centering
        \includegraphics[width=0.4285\textwidth]{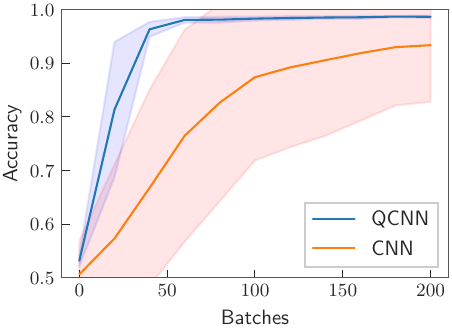}
    }
    }
        \hfill
        \hspace*{-0.65em}
    {
    \captionsetup{oneside,margin={1cm,0cm}}
    \subfloat[Test loss]{
            \includegraphics[width=0.4285\textwidth]{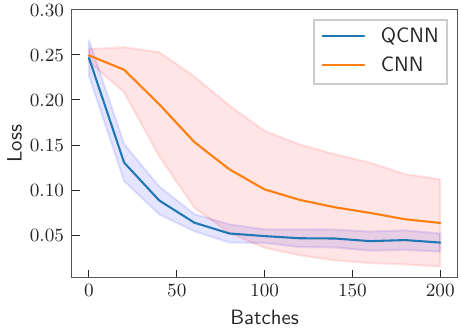}
    }
    }
    \caption{Performance comparison of the 49-qubit regular QCNN model using WUE embedding with $R_x$ gates in the second encoding layer, followed by a QCNN with alternating interpolation and pooling. In the pooling circuit, the $CRZ$ gate was replaced with a $CZ$ gate followed by an $R_z$ gate. Subfigure (a) presents the test accuracy, while subfigure (b) displays the test loss. The shaded regions indicate ±1 standard deviation from the mean across experimental runs.}
        \label{fig:rqcnn_wue_rx_cz_rz}
\end{figure*}

\onecolumn

\begin{table}[ht]
    \centering
    \begin{tabular}{c C C C C}
        \toprule
        Embedding & 1\times1 \text{ Conv} & \text{Parameters} & \text{Accuracy} & \text{Loss} \\
        \midrule
        QE & 0 & 20 & 97.5 \pm 1.2 & 0.05 \pm 0.0 \\ 
        QE & 1 & 23 & 92.7 \pm 4.7 & 0.08 \pm 0.03 \\ 
        QE & 2 & 26 & 96.0 \pm 1.4 & 0.06 \pm 0.01 \\ 
        QE & 3 & 29 & 95.6 \pm 2.0 & 0.06 \pm 0.01 \\ 
        DQE & 0 & 14 & 86.0 \pm 6.9 & 0.12 \pm 0.02 \\ 
        DQE & 1 & 17 & 82.3 \pm 2.3 & 0.13 \pm 0.02 \\ 
        DQE & 2 & 20 & 80.5 \pm 8.2 & 0.13 \pm 0.03 \\ 
        DQE & 3 & 14 & 86.0 \pm 6.9 & 0.1 \pm 0.0 \\ 
        WUE & 0 & 18 & 98.4 \pm 0.2 & 0.04 \pm 0.0 \\ 
        WUE & 1 & 21 & 97.9 \pm 0.1 & 0.04 \pm 0.0 \\ 
        WUE & 2 & 24 & 98.2 \pm 0.2 & 0.03 \pm 0.0 \\ 
        WUE & 3 & 27 & 98.4 \pm 0.1 & 0.03 \pm 0.0 \\ 
        UE & 0 & 17 & 95.0 \pm 1.9 & 0.06 \pm 0.01 \\ 
        UE & 1 & 20 & 96.3 \pm 0.7 & 0.06 \pm 0.0 \\ 
        UE & 2 & 23 & 94.9 \pm 0.9 & 0.07 \pm 0.01 \\ 
        UE & 3 & 26 & 96.0 \pm 0.5 & 0.06 \pm 0.0 \\ 
        CNN & - & 22 & 93.4\pm10.6 & 0.06 \pm 0.05\\
        \bottomrule
    \end{tabular}
    \caption{
    Performance of the QCNN using a fragment encoding with a single layer of $U3$ gates. The table compares various quantum embeddings and the number of $1\times1$ convolutions applied with $U3$ gates from the last layer to the first. It shows parameter count, test accuracy, and test loss for classifying digits 0 and 1. The CNN baseline from the main text is included for reference.
    }
    \label{tab:rqcnn_q49_embedding_01_compact}
\end{table}

\twocolumn

\end{document}


\title{Supplementary\\
Efficient Quantum Convolutional Neural Networks for Image Classification: Overcoming Hardware Constraints}

\author{Peter Röseler}
\email{p.roeseler@fz-juelich.de}
\affiliation{Department of Computer Science, University of Bonn, 53111 Bonn, Germany}
\affiliation{Bayer AG, Kaiser-Wilhelm-Allee 1
51373 Leverkusen, Germany}
\affiliation{Jülich Supercomputing Centre, Institute for Advanced Simulation, Forschungszentrum Jülich, 52425 Jülich, Germany}

\author{Oliver Schaudt}
\affiliation{Bayer AG, Kaiser-Wilhelm-Allee 1
51373 Leverkusen, Germany}

\author{Helmut Berg}
\affiliation{Bayer AG, Kaiser-Wilhelm-Allee 1
51373 Leverkusen, Germany}

\author{Christian Bauckhage}
\affiliation{Department of Computer Science, University of Bonn, 53111 Bonn, Germany}

\author{Matthias Koch}
\affiliation{Bayer AG, Kaiser-Wilhelm-Allee 1
51373 Leverkusen, Germany}


\maketitle

\section{Classical fragment encoding}    

    As a side note, the fragment encoding procedure can also yield an exponential decrease in memory usage for classical CNNs. Once the segment of a receptive field in a given layer is processed, the memory used for that computation can be released. In common scenarios, this can lead to an exponential decrease in memory usage, given continuous access to an external storage system that provides the input data on demand Proposition \ref{prp:multi-depth_kernel_processing}. 
    \begin{prop}
        \label{prp:multi-depth_kernel_processing}
        Let $n\times n$ be the size of the input with $n\in\mathbb{N}$. For a CNN with altering layers of convolution and pooling that have fixed kernel sizes of $k\times k$ with stride 1 and $m\times m$ with stride $m$ respectively, where $m,k\in\mathbb{N}$, given that $\exists c\in\mathbb{N}: m^c=n:$
        \begin{equation}
            1+\sum_{d=1}^{log_m(n)}(k^2+m^2-2)=1+(k^2+m^2-2)log_m(n)
        \end{equation}
        values needs to be stored at most.
    \end{prop}
    \begin{proof}
            \begin{enumerate}
                \item Consider the first layer. A single pooling kernel requires $m^2$ many inputs
                \\
                $\implies m^2$ convolutions of size $k^2$ are needed.  After computing a convolution, the memory of the input can be freed.
                \\
                $\implies$ At most $m^2 + k^2 - 1$ values need to be stored to compute a single output of the next pooling layer.
                \item Consider the next convolutional layer. Using $m^2+k^2-1$ values, a single input for the current layer can be generated as described in (i).
                \\
                $\implies$ A convolution requires $k^2$ inputs, so at most $k^2 - 1 + m^2 + k^2 - 1$ values are needed.
                \\
                $\implies$ At most $m^2 - 1 + k^2 - 1 + m^2 + k^2 - 1$ values are required for an output of the second pooling layer.
                \item By using the same arguments as in (ii) and considering the number of values required for the first layer in (i), an output in the $j$-th pooling layer can be generated with:
                \begin{equation}
                    k^2+m^2-1+\sum_{d=1}^{j-1}(k^2+m^2-2)
                \end{equation}
                values. The pooling operation is the only step that reduces the image size, which decreases by a factor of $m$ each time. This leads to a size of $\frac{n}{m^j}\times \frac{n}{m^j}$ in the $j$-th layer. Given that $\exists c\in\mathbb{N}$ such that $m^c = n$, the total number of values that need to be stored until the CNN reaches a size of 1 is:
                \begin{equation}
                    1+\sum_{d=1}^{log_m(n)}(k^2+m^2-2)=1+(k^2+m^2-2)\log_m(n).
                \end{equation}
            \end{enumerate}
    \end{proof}

\clearpage

\section{CNN baselines}   
    \begin{table}[ht]  
    \centering
    \renewcommand{\arraystretch}{1.5} 
    \setlength{\tabcolsep}{12pt} 
    \begin{tabular}{c c c} 
    \hline
    Layer & Filter Size & \makecell{Activation\\Function} \\ 
    \hline
    Conv2D-1 & $3\times3\times1$ & ELU \\ 
    MaxP2D-1 & $2\times2$ & - \\ 
    Conv2D-2 & $2\times2\times1$ & ELU \\ 
    MaxP2D-2 & $2\times2$ & - \\ 
    Conv2D-3 & $3\times3\times1$ & ELU \\ 
    MaxP2D-3 & $2\times2$ & - \\ 
    Conv2D-4 & $1\times1\times1$ & ELU \\ 
    MaxP2D-4 & $2\times2$ & Sigmoid \\ 
    \hline
    \end{tabular}
    \caption{Optimized CNN for classes 7 and 8 with $27$ parameters.}
    \label{tab:network_architecture_78}
    \end{table}

    \begin{table}[ht]
    \centering
    \renewcommand{\arraystretch}{1.5} 
    \setlength{\tabcolsep}{12pt} 
    \begin{tabular}{c c c} 
    \hline
    Layer & Filter Size & \makecell{Activation\\Function} \\ 
    \hline
    Conv2D-1 & $3\times3\times3$ & ELU \\ 
    MaxP2D-1 & $2\times2$ & - \\
    Conv2D-2 & $2\times2\times1$ & ELU \\
    Conv2D-3 & $3\times3\times4$ & ELU \\ 
    MaxP2D-2 & $2\times2$ & - \\ 
    Conv2D-4 & $3\times3\times4$ & ELU \\
    Conv2D-5 & $3\times3\times1$ & Sigmoid \\
    \hline
    \end{tabular}
    \caption{Optimized CNN for classes 0, 1, 2, and 3 with $263$ parameters.}
    \label{tab:network_architecture_0123}
    \end{table}

    \begin{table}[ht]
    \centering
    \renewcommand{\arraystretch}{1.5} 
    \setlength{\tabcolsep}{12pt} 
    \begin{tabular}{c c c} 
    \hline
    Layer & Filter Size & \makecell{Activation\\Function} \\ 
    \hline
    Conv2D-1 & $3\times3\times5$ & ELU \\ 
    MaxP2D-1 & $2\times2$ & - \\ 
    Conv2D-2 & $1\times1\times1$ & ELU \\ 
    Conv2D-3 & $2\times2\times4$ & ELU \\ 
    MaxP2D-2 & $2\times2$ & - \\  
    Conv2D-4 & $3\times3\times5^*$ & ELU \\ 
    MaxP2D-3 & $2\times2$ & - \\ 
    Conv2D-5 & $3\times3\times1$ & Sigmoid  \\
    \hline
    \end{tabular}
    \caption{Optimized CNN for classes greater than 4 with $307$ parameters. Layers marked with $^*$ indicate padding set to 1.}
    \label{tab:network_architecture_g4}
    \end{table}

\clearpage

\section{Baseline circuits}
    \label{app:baseline_circuits}
    This section presents the quantum circuits used in the ansatz search experiments, adapted from previous research as baselines. These circuits were adjusted for 3- and 4-qubit configurations to match the experimental requirements in this work. Figures \ref{fig:convolutional_circuits_3q} and \ref{fig:quantum_circuits_4q} depict the circuits for 3 and 4 qubits, respectively. Circuits for larger systems up to 9 qubits can be extrapolated accordingly, following the same structural patterns.
     \begin{figure}[ht]
        \centering
            \begin{subfigure}{0.32\textwidth}
                \centering
                
                \begin{quantikz}[rounded corners, column sep=5pt, row sep=4pt]
                & \gate[][0em][1.75em]{H} & \qw & \ctrl{1} &  \gate{R_x(\theta_0)} & \qw \\
                & \gate[][0em][1.75em]{H} & \ctrl{1} & \ctrl{-1} & \gate{R_x(\theta_1)} & \qw \\
                & \gate[][0em][1.75em]{H} & \ctrl{-1} & \qw & \gate{R_x(\theta_2)} & \qw
                \end{quantikz}
                \caption{Circuit 1}
            \end{subfigure}
            \begin{subfigure}{0.32\textwidth}
                \centering
                
                \begin{quantikz}[rounded corners, column sep=5pt, row sep=4pt]
                & \gate{R_y(\theta_0)} & \ctrl{1} & \qw & \qw & \qw \\
                & \gate{R_y(\theta_1)} & \targ{} & \ctrl{1} & \qw & \qw \\
                & \gate{R_y(\theta_2)} & \qw & \targ{} & \gate{R_y(\theta_3)} & \qw
                \end{quantikz}
                \caption{Circuit 2}
            \end{subfigure}
            \vspace{0.5cm}
            \begin{subfigure}{0.32\textwidth}
                \centering
                
                \begin{quantikz}[rounded corners, column sep=5pt, row sep=4pt]
                & \gate{R_x(\theta_0)} & \gate{R_z(\theta_3)} & \qw & \targ{} & \qw \\
                & \gate{R_x(\theta_1)} & \gate{R_z(\theta_4)} & \targ{} & \ctrl{-1} & \qw\\
                & \gate{R_x(\theta_2)} & \gate{R_z(\theta_5)} & \ctrl{-1} & \qw & \qw
                \end{quantikz}
                \caption{Circuit 3}
            \end{subfigure}
            \vspace{0.5cm}
            \begin{subfigure}{0.48\textwidth}
                \centering
                
                \begin{quantikz}[rounded corners, column sep=5pt, row sep=4pt]
                & \gate{R_y(\theta_0)} & \qw & \ctrl{1} & \ctrl{2} & \gate{R_y(\theta_3)} & \qw \\
                & \gate{R_y(\theta_1)} & \ctrl{1} & \ctrl{-1} & \qw & \gate{R_y(\theta_4)} & \qw \\
                & \gate{R_y(\theta_2)} & \ctrl{-1} & \qw & \ctrl{-2} & \gate{R_y(\theta_5)} & \qw
                \end{quantikz}
                \caption{Circuit 4}
            \end{subfigure}
            \hspace{-0.2cm} 
            \begin{subfigure}{0.48\textwidth}
                \centering
                
                \begin{quantikz}[rounded corners, column sep=5pt, row sep=4pt]
                & \gate{R_y(\theta_0)} & \targ{} & \qw & \ctrl{1} & \gate{R_y(\theta_3)} & \qw & \ctrl{2} & \targ{} & \qw \\
                & \gate{R_y(\theta_1)} & \qw & \ctrl{1} & \targ{} & \gate{R_y(\theta_4)} & \targ{} & \qw & \ctrl{-1} & \qw \\
                & \gate{R_y(\theta_2)} & \ctrl{-2} & \targ{} & \qw & \gate{R_y(\theta_5)} & \ctrl{-1} & \targ{} & \qw & \qw
                \end{quantikz}
                \caption{Circuit 5}
            \end{subfigure}
            \vspace{0.5cm}
            \begin{subfigure}{0.98\textwidth}
                \centering
                
                \begin{quantikz}[rounded corners, column sep=5pt, row sep=4pt]
                & \gate{R_y(\theta_0)} & \gate{R_z(\theta_3)} & \qw & \ctrl{1} & \gate{R_y(\theta_6)} & \qw & \ctrl{2} & \gate{R_z(\theta_{11})} & \qw \\
                & \gate{R_y(\theta_1)} & \qw & \ctrl{1} & \gate{R_z(\theta_5)} & \gate{R_y(\theta_7)} & \gate{R_z(\theta_{9})} & \qw & \ctrl{-1} & \qw \\
                & \gate{R_y(\theta_2)} & \ctrl{-2} & \gate{R_z(\theta_4)} & \qw & \gate{R_y(\theta_8)} & \ctrl{-1} & \gate{R_z(\theta_{10})} & \qw & \qw
                \end{quantikz}
                \caption{Circuit 6}
            \end{subfigure}
        
        \caption{Parameterized quantum circuits for a 3-qubit convolution layer. The circuits are adapted from previous research papers.}
        \label{fig:convolutional_circuits_3q}
        \end{figure}

        \begin{figure}[ht]
        \centering
        \begin{subfigure}{0.48\textwidth}
            \centering
            
            \begin{quantikz}[rounded corners, column sep=5pt, row sep=4pt]
            & \gate[][0em][1.75em]{H} & \qw & \qw & \ctrl{1} & \qw & \qw & \gate{R_x(\theta_0)} & \qw \\
            & \gate[][0em][1.75em]{H} & \qw & \ctrl{1} & \ctrl{-1} & \qw & \qw & \gate{R_x(\theta_1)} & \qw \\
            & \gate[][0em][1.75em]{H} & \ctrl{1} & \ctrl{-1} & \qw & \qw & \qw & \gate{R_x(\theta_2)} & \qw \\
            & \gate[][0em][1.75em]{H} & \ctrl{-1} & \qw & \qw & \qw & \qw & \gate{R_x(\theta_3)} & \qw
            \end{quantikz}
            \caption{Circuit 1}
        \end{subfigure}
        \begin{subfigure}{0.48\textwidth}
            \centering
            
            \begin{quantikz}[rounded corners, column sep=5pt, row sep=4pt]
            & \gate{R_y(\theta_0)} & \ctrl{1} & \qw  & \qw & \qw & \qw\\
            & \gate{R_y(\theta_1)} & \targ{} & \gate{R_y(\theta_4)} & \ctrl{2} & \qw & \qw\\
            & \gate{R_y(\theta_2)} & \ctrl{1} & \qw & \qw & \qw & \qw\\
            & \gate{R_y(\theta_3)} & \targ{} & \gate{R_y(\theta_5)} & \targ{} & \gate{R_y(\theta_6)} & \qw
            \end{quantikz}
            \caption{Circuit 2}
        \end{subfigure}
        \vspace{0.5cm}
        \begin{subfigure}{0.48\textwidth}
            \centering
            
            \begin{quantikz}[rounded corners, column sep=5pt, row sep=4pt]
            & \gate{R_x(\theta_0)} & \gate{R_z(\theta_4)} & \qw & \qw & \targ{} & \qw \\
            & \gate{R_x(\theta_1)} & \gate{R_z(\theta_5)} & \qw & \targ{} & \ctrl{-1} & \qw \\
            & \gate{R_x(\theta_2)} & \gate{R_z(\theta_6)} & \targ{} & \ctrl{-1} & \qw & \qw \\
            & \gate{R_x(\theta_3)} & \gate{R_z(\theta_7)} & \ctrl{-1} & \qw & \qw & \qw
            \end{quantikz}
            \caption{Circuit 3}
        \end{subfigure}
        \begin{subfigure}{0.48\textwidth}
            \centering
            
            \begin{quantikz}[rounded corners, column sep=5pt, row sep=4pt]
            & \gate{R_y(\theta_0)} & \qw & \qw & \ctrl{1} & \ctrl{3} & \gate{R_y(\theta_4)} & \qw \\
            & \gate{R_y(\theta_1)} & \qw & \ctrl{1} & \ctrl{-1} & \qw &\gate{R_y(\theta_5)} & \qw \\
            & \gate{R_y(\theta_2)} & \ctrl{1} & \ctrl{-1} & \qw & \qw & \gate{R_y(\theta_6)} & \qw \\
            & \gate{R_y(\theta_3)} & \ctrl{-1} & \qw & \qw & \ctrl{-3} & \gate{R_y(\theta_7)} & \qw
            \end{quantikz}
            \caption{Circuit 4}
        \end{subfigure}
        \vspace{0.5cm}
        \begin{subfigure}{0.48\textwidth}
            \centering
            
            \begin{quantikz}[rounded corners, column sep=5pt, row sep=4pt]
            & \gate{R_y(\theta_0)} & \targ{} & \qw & \qw & \ctrl{1} & \gate{R_y(\theta_4)} & \qw & \ctrl{3} & \targ{} & \qw & \qw & \qw\\
            & \gate{R_y(\theta_1)} & \qw & \qw & \ctrl{1} & \targ{} & \gate{R_y(\theta_5)} & \qw & \qw & \ctrl{-1} & \qw & \targ{} & \qw \\
            & \gate{R_y(\theta_2)} & \qw & \ctrl{1} & \targ{} & \qw & \gate{R_y(\theta_6)} & \targ{} & \qw & \qw & \qw & \ctrl{-1} & \qw \\
            & \gate{R_y(\theta_3)} & \ctrl{-3} & \targ{} & \qw & \qw & \gate{R_y(\theta_7)} & \ctrl{-1} & \targ{} & \qw & \qw & \qw & \qw
            \end{quantikz}
            \caption{Circuit 5}
        \end{subfigure}
    \caption{Parameterized quantum circuits for a 4-qubit convolution layer. The circuits are adapted from previous research papers.}
    \label{fig:quantum_circuits_4q}
    \end{figure}

\clearpage

\section{Ansatz search}
The ansatz search yielded circuits with the lowest objective function values, indicating best overall performance for both architectures. For hybrid QCNN, the discovered circuits achieved better balance between expressibility and entanglement while being more compact than previous designs. Regular QCNN circuits showed stronger emphasis on entanglement, likely due to larger expressibility normalization. Both architectures relied on parameter sharing, which appears crucial for circuit efficiency.

While the search methodology appears adaptable across quantum architectures, circuit performance was architecture-dependent, similar to the 'no free lunch' theorem in classical ML. For instance, circuits 1 and 3 showed substantial changes in expressibility when transferred between architectures. Additionally, the intialization via random statevector or parameters for the encoding circuits can significantly impact the entanglement values, as seen in circuit 1. Beyond architectural dependencies, another aspect to consider for future research is that the exponential growth of the state space with qubit count presents challenges for scaling.

        \begin{table}[ht]
        \centering
        \renewcommand{\arraystretch}{2.5}
        \begin{adjustbox}{max width=\textwidth}
        \begin{tabular}{|c|c|c|c|c|c|}
        \hline
        \textbf{Name} & \textbf{Circuit} & \textbf{Name} & \textbf{Circuit} & \textbf{Name} & \textbf{Circuit} \\ \hline
        
        $H$ & 
        \begin{quantikz} \qw & \gate{H} & \qw \end{quantikz} & 
        $SX$ & 
        \begin{quantikz} \qw & \gate{\surd X} & \qw \end{quantikz} & 
        $X (NOT)$ & 
        \begin{quantikz} \qw & \gate{X} & \qw \end{quantikz} \\ \hline
        $Y$ & 
        \begin{quantikz} \qw & \gate{Y} & \qw \end{quantikz} & 
        $Z$ & 
        \begin{quantikz} \qw & \gate{Z} & \qw \end{quantikz} & 
        $R_x(\theta)$ & 
        \begin{quantikz} \qw & \gate{R_x(\theta)} & \qw \end{quantikz} \\ \hline
        $R_y(\theta)$ & 
        \begin{quantikz} \qw & \gate{R_y(\theta)} & \qw \end{quantikz} & 
        $R_z(\theta)$ & 
        \begin{quantikz} \qw & \gate{R_z(\theta)} & \qw \end{quantikz} & 
        $U3(\theta,\phi,\lambda)$ & 
        \begin{quantikz} \qw & \gate{U3(\theta,\phi,\lambda)} & \qw \end{quantikz} \\ \hline
        $CX (CNOT)$ & 
        \begin{quantikz} \qw & \ctrl{1} & \qw  \\ \qw & \targ & \qw & \qw \end{quantikz} & 
        $CY$ & 
        \begin{quantikz} \qw & \ctrl{1} & \qw \\ \qw & \gate{Y}  & \qw  \end{quantikz} & 
        $CZ$ & 
        \begin{quantikz} \qw & \ctrl{1} & \qw \\ \qw & \ctrl{-1} & \qw \end{quantikz} \\ \hline
        $ECR$ & 
        \begin{quantikz} \qw & \gate[2]{ECR} \gateinput{0} & \qw \\ \qw & \gateinput{1} & \qw \end{quantikz} & 
        $CRX(\theta)$ & 
        \begin{quantikz} \qw & \ctrl{1} & \qw \\ \qw & \gate{R_x(\theta)} & \qw \end{quantikz} & 
        $CRY(\theta)$ & 
        \begin{quantikz} \qw & \ctrl{1} & \qw \\ \qw & \gate{R_y(\theta)} & \qw \end{quantikz} \\ \hline
        $CRZ(\theta)$ & 
        \begin{quantikz} \qw & \ctrl{1} & \qw \\ \qw & \gate{R_z(\theta)} & \qw \end{quantikz} & & & & \\ \hline
        \end{tabular}
        \end{adjustbox}
        \caption{Quantum gates used for the ansatz search.}
        \label{tab:quantum_circuits_search}
        \end{table}

\clearpage

\begin{figure}[ht]
    \centering
    \begin{quantikz}[rounded corners, column sep=5pt, row sep=4pt]
    & \gate{Z} & \ctrl{1} & \gate{R_x(\theta_0)} & \qw & \qw \\
    & \gate{R_y(\theta_0)} & \ctrl{-1} & \ctrl{-1} & \gate{R_y(\theta_1)} & \qw
    \end{quantikz}
    \caption{2-qubit circuit from the ansatz search of the hybrid QCNN.}
    \label{fig:mqcnn_pqc_search_circ_2q}
\end{figure}

\begin{table}[ht]   
    \centering
    \resizebox{\textwidth}{!}{  
    \begin{tabular}{c C C C C C C}
        \toprule
        Circuit ID & \text{Parameters} & \text{Depth} & \text{Gates} & \text{Expressibility} & \text{Entanglement} & \mathcal{L}_{PQC} \\
        \midrule
        1 & \mathbf{2} & \mathbf{3} & 5 & 0.252 \pm 0.31  & 0.203 \pm 0.258 & 1.664 \\
        2 & 3 & \mathbf{3} & \mathbf{4} & 0.005 \pm 0.001 & 0.249 \pm 0.003 & 1.376 \\
        3 & 4 & \mathbf{3} & 5 & 0.331 \pm 0.325 & \mathbf{0.501 \pm 0.14} & 1.23 \\
        4 & 4 & \mathbf{3} & 5 & 0.005 \pm 0.001 & 0.25 \pm 0.002 & 1.375 \\
        5 & 4 & 4 & 6 & \mathbf{0.002 \pm 0.001} & 0.313 \pm 0.002 & 1.218 \\
        6 & 6 & 4 & 6  & 0.005 \pm 0.001 & 0.214 \pm 0.002 & 1.466 \\
        AS & \mathbf{2} & 4 & 5 & 0.013 \pm 0.012 & 0.389 \pm 0.09 & \mathbf{1.028} \\
        \midrule
        \midrule
        \makecell{ Thresholds } & \leq2 & \leq6 & \leq10 & \leq0.016 & \geq0.4 & \leq1.0 \\
        \bottomrule
    \end{tabular}
    }
    \caption{Evaluation of PQCs for the hybrid QCNN architecture with 2 qubits. The table details each circuit's complexity, expressibility and entanglement (mean $\pm$ standard deviation), and objective function value ($\mathcal{L}_{PQC}$, mean). Bold numbers indicate the best entries in each column. The bottom row (Thresholds) shows the constraints and objectives of the ansatz search (AS).}
    \label{tab:mqcnn_pqc_2q}
\end{table}

\clearpage

\begin{figure}[ht]
    \centering
    \begin{quantikz}[rounded corners, column sep=5pt, row sep=4pt]
    & \qw & \ctrl{2} & \gate{Y} & \qw & \qw &  \\
    & \gate{\surd X} & \qw & \ctrl{-1} & \qw  & \qw\\
    & \qw & \targ{} & \gate{R_y(\theta_0)} & \gate{U3(\theta_1, \theta_0, \theta_2)} & \qw
    \end{quantikz}
    \caption{3-qubit circuit from the ansatz search of the hybrid QCNN.}
    \label{fig:mqcnn_pqc_search_circ_3q}
\end{figure}

\begin{table}[ht]
    \centering
    \resizebox{\textwidth}{!}{  
    \begin{tabular}{c C C C C C C}
        \toprule
        Circuit ID & \text{Parameters} & \text{Depth} & \text{Gates} & \text{Expressibility} & \text{Entanglement} & \mathcal{L}_{PQC} \\
        \midrule
        1 & \mathbf{3} & 4 & 8 & 0.245 \pm 0.304  &  0.324 \pm 0.252 & 1.681 \\
        2 & 4 & 4 & 6 & 0.002 \pm 0.001  &  0.375 \pm 0.003 & 1.438 \\
        3 & 6 & 4 & 8 & 0.241 \pm 0.305  &  0.606 \pm 0.08AS & 1.255 \\
        4 & 6 & 5 & 9 & 0.002 \pm 0.0  &  0.375 \pm 0.002 & 1.438 \\
        5 & 6 & 8 & 12 & 0.253 \pm 0.006  &  0.626 \pm 0.002 & 1.234 \\
        6 & 12 & 8 & 12 & \mathbf{0.0 \pm 0.0}  &  0.397 \pm 0.002 & 1.405 \\
        AS & \mathbf{3} & \mathbf{3} & \mathbf{5} & 0.11 \pm 0.223  &  \mathbf{0.758 \pm 0.1} & \mathbf{1.069} \\
        \midrule
        \midrule
        \makecell{ Thresholds } & \leq3 & \leq9 & \leq15 & \leq0.016 & \geq0.667 & \leq1.0 \\
        \bottomrule
    \end{tabular}
    }
    \caption{Evaluation of PQCs for the hybrid QCNN architecture with 3 qubits. The table details each circuit's complexity, expressibility and entanglement (mean $\pm$ standard deviation), and objective function value ($\mathcal{L}_{PQC}$, mean). Bold numbers indicate the best entries in each column.  The bottom row (Thresholds) shows the constraints and objectives of the ansatz search (AS).}
    \label{tab:mqcnn_pqc_3q}
\end{table}

\clearpage

\begin{figure}[ht]
    \centering
    \begin{quantikz}[rounded corners, column sep=5pt, row sep=4pt]
    & \ctrl{1} & \gate{\surd X} & \qw & \ctrl{2} & \qw & \qw \\
    & \gate{R_y(\theta_0)} & \ctrl{1} & \gate{Y} & \qw & \qw & \qw \\
    & \qw & \gate{R_y(\theta_0)} & \qw & \gate{Y} & \qw & \qw\\
    & \gate{U3(\theta_0, \theta_0, \theta_0)} & \gate{R_y(\theta_1)} & \ctrl{-2} & \qw & \qw & \qw
    \end{quantikz}
    \caption{4-qubit circuit from the ansatz search of the hybrid QCNN.}
    \label{fig:mqcnn_pqc_search_circ_4q}
\end{figure}

\begin{table}[ht]
    \centering
    \resizebox{\textwidth}{!}{  
    \begin{tabular}{c C C C C C C}
        \toprule
        Circuit ID & \text{Parameters} & \text{Depth} & \text{Gates} & \text{Expressibility} & \text{Entanglement} & \mathcal{L}_{PQC} \\
        \midrule
        1 & 4 & 5 & 11 & 0.252 \pm 0.306  &  0.319 \pm 0.219 & 1.784 \\
        2 & 7 & 5 & 10 & 0.021 \pm 0.002  &  0.367 \pm 0.002 & 1.558 \\
        3 & 8 & 5 & 11 & 0.249 \pm 0.307  &  0.677 \pm 0.062 & 1.348 \\
        4 & 8 & 6 & 12 & \mathbf{0.002 \pm 0.001}  &  0.375 \pm 0.002 & 1.545 \\
        5 & 8 & 9 & 16 & 0.348 \pm 0.005  &  0.711 \pm 0.001 & 1.379 \\
        AS & \mathbf{2} & \mathbf{3} & \mathbf{7} & 0.006 \pm 0.005  &  \mathbf{0.859 \pm 0.017} & \mathbf{0.37} \\
        \midrule
        \midrule
        \makecell{ Thresholds } & \leq4 & \leq12 & \leq20 & \leq0.016 & \geq0.824 & \leq1.0 \\
        \bottomrule
    \end{tabular}
    }
    \caption{Evaluation of PQCs for the hybrid QCNN architecture with 4 qubits. The table details each circuit's complexity, expressibility and entanglement (mean $\pm$ standard deviation), and objective function value ($\mathcal{L}_{PQC}$, mean). Bold numbers indicate the best entries in each column.  The bottom row (Thresholds) shows the constraints and objectives of the ansatz search (AS).}
    \label{tab:mqcnn_pqc_4q}
\end{table}

\clearpage

\begin{figure}[ht]
    \centering
    \resizebox{.97\textwidth}{!}{
    \begin{quantikz}[rounded corners, column sep=5pt, row sep=4pt]
    & \gate{U3(\theta_3, \theta_4, \theta_5)} & \qw & \qw & \qw & \qw & \qw & \qw & \gate{X} & \qw & \ctrl{7} & \gate[4]{ECR} \gateinput{0} & \qw \\
    & \qw & \qw & \gate{Y} & \ctrl{3} & \qw & \qw & \qw & \qw & \qw & \qw & \qw & \qw \\
    & \qw & \qw & \qw & \qw & \gate[3]{ECR} \gateinput{1} & \ctrl{4} & \gate[3]{ECR} \gateinput{1} & \qw & \qw & \qw & \qw & \qw \\
    & \gate[3]{ECR} \gateinput{0} & \qw & \qw & \qw & \qw & \qw & \qw & \qw & \qw & \qw & \gateinput{1} & \qw \\
    & \qw & \gate{X} & \qw & \gate{Y} & \gateinput{0} & \qw & \gateinput{0} & \qw & \ctrl{2} & \qw & \gate{R_z(\theta_6)} & \qw \\
    & \gateinput{1} & \qw & \ctrl{-4} & \gate{R_z(\theta_6)} & \gate{U3(\theta_7, \theta_8, \theta_5)} & \qw & \qw & \qw & \qw & \qw & \qw & \qw \\
    & \qw & \qw & \qw & \qw & \qw & \gate{Y} & \ctrl{2} & \ctrl{-6} & \gate{R_x(\theta_5)} & \qw & \qw & \qw \\
    & \gate{Z} & \qw & \qw & \qw & \qw & \qw & \qw & \qw & \qw & \gate{Y} & \qw & \qw \\
    & \gate{X} & \gate{U3(\theta_0, \theta_1, \theta_2)} & \gate{R_x(\theta_4)} & \qw & \qw & \qw & \ctrl{-2} & \qw & \qw & \qw & \qw & 
    \end{quantikz}
    }
    \caption{9-qubit circuit from the ansatz search of the hybrid QCNN.}
    \label{fig:mqcnn_pqc_search_circ_9q}
\end{figure}

\begin{table}[ht]
    \centering
    \resizebox{\textwidth}{!}{  
    \begin{tabular}{c C C C C C C}
        \toprule
        Circuit ID & \text{Parameters} & \text{Depth} & \text{Gates} & \text{Expressibility} & \text{Entanglement} & \mathcal{L}_{PQC} \\
        \midrule
        1 & \mathbf{9} & 10 & 26 & 0.25 \pm 0.306  &  0.371 \pm 0.153 & 1.798 \\
        2 & 16 & \mathbf{8} & 24 & 0.204 \pm 0.005  &  0.435 \pm 0.001 & 1.7 \\
        3 & 18 & 10 & 26 & 0.232 \pm 0.299  &  0.776 \pm 0.031 & 1.376 \\
        4 & 18 & 11 & 27 & \mathbf{0.002 \pm 0.001}  &  0.375 \pm 0.001 & 1.623 \\
        5 & 18 & 13 & 36 & 0.421 \pm 0.004  &  0.763 \pm 0.001 & 1.528 \\
        AS & \mathbf{9} & 9 & \mathbf{20} & \mathbf{0.002 \pm 0.001}  & \mathbf{0.962 \pm 0.002} & \mathbf{1.032} \\
        \midrule
        \midrule
        \makecell{ Thresholds } & \leq9 & \leq27 & \leq45 & \leq0.016 & \geq0.994 & \leq1.0 \\
        \bottomrule
    \end{tabular}
    }
    \caption{Evaluation of PQCs for the hybrid QCNN architecture with 9 qubits. The table details each circuit's complexity, expressibility and entanglement (mean $\pm$ standard deviation), and objective function value ($\mathcal{L}_{PQC}$, mean). Bold numbers indicate the best entries in each column.  The bottom row (Thresholds) shows the constraints and objectives of the ansatz search (AS).}
    \label{tab:mqcnn_pqc_9q}
\end{table}

\clearpage

\begin{figure}[ht]
    \centering
    \begin{quantikz}[rounded corners, column sep=5pt, row sep=4pt]
    & \gate{R_z(\theta_0)} & \ctrl{1} & \gate{U3(\theta_0,\theta_0,\theta_0)} & \qw & \gate[2]{ECR} \gateinput{1} & \gate{U3(\theta_1,\theta_0,\theta_0)} & \qw \\
    & \gate{Y} & \gate{R_z(\theta_0)} & \gate{R_y(\theta_0)} & \gate{R_y(\theta_0)} & \gateinput{0} & \qw & \qw 
    \end{quantikz}
    \caption{2-qubit circuit from the ansatz search of the regular QCNN.}
    \label{fig:rqcnn_pqc_search_circ_2q}
\end{figure}

\begin{table}[ht]
    \centering
    \resizebox{\textwidth}{!}{  
    \begin{tabular}{c C C C C C C}
        \toprule
        Circuit ID & \text{Parameters} & \text{Depth} & \text{Gates} & \text{Expressibility} & \text{Entanglement} & \mathcal{L}_{PQC} \\
        \midrule
        1 & \mathbf{2} & \mathbf{3} & 5 & 0.734 \pm 0.758  &  \mathbf{1.0 \pm 0.0} & 1.055 \\
        2 & 3 & \mathbf{3} & \mathbf{4} & 0.127 \pm 0.137  &  0.252 \pm 0.001 & 1.378 \\
        3 & 4 & \mathbf{3} & 5 & 0.035 \pm 0.033  &  0.371 \pm 0.001 & 1.073 \\
        4 & 4 & \mathbf{3} & 5 & 0.035 \pm 0.033  &  0.251 \pm 0.001 & 1.373 \\
        5 & 4 & 4 & 6 & 0.033 \pm 0.032  &  0.318 \pm 0.006 & 1.206 \\
        6 & 6 & 4 & 6 & \mathbf{0.011 \pm 0.005}  &  0.212 \pm 0.003 & 1.471 \\
        AS & \mathbf{2} & 6 & 8 & 0.096 \pm 0.043  &  0.407 \pm 0.001 & \mathbf{1.006} \\
        \midrule
        \midrule
        \makecell{ Thresholds } & \leq2 & \leq6 & \leq10 & \leq0.021 & \geq0.4 & \leq1.0 \\
        \bottomrule
    \end{tabular}
    }
    \caption{Evaluation of PQCs for the regular QCNN architecture with 2 qubits. The table details each circuit's complexity, expressibility and entanglement (mean $\pm$ standard deviation), and objective function value ($\mathcal{L}_{PQC}$, mean). Bold numbers indicate the best entries in each column.  The bottom row (Thresholds) shows the constraints and objectives of the ansatz search (AS).}
    \label{tab:rqcnn_pqc_2q}
\end{table}


\clearpage

\begin{figure}[ht]
    \centering
    \begin{quantikz}[rounded corners, column sep=5pt, row sep=4pt]
    & \gate{R_x(\theta_0)} & \gate{U3(\theta_0,\theta_0,\theta_0)} & \gate{R_y(\theta_0)} & \ctrl{1} & \qw & \gate{R_x(\theta_1)} & \qw &  \qw &  \qw \\
    & \gate{H} & \qw & \qw & \gate{Y} & \gate[2]{ECR} \gateinput{0}  & \qw & \qw & \gate{U3(\theta_2,\theta_1,\theta_2)} &  \qw \\
    & \gate{R_z(\theta_1)} & \gate{R_z(\theta_1)} & \qw & \qw & \gateinput{1} & \ctrl{-2} & \qw & \qw&  \qw 
    \end{quantikz}
    \caption{
    3-qubit circuit from the ansatz search of the regular QCNN.
    }
    \label{fig:rqcnn_pqc_search_circ_3q}
\end{figure}

\begin{table}[ht]
    \centering
    \resizebox{\textwidth}{!}{  
    \begin{tabular}{c C C C C C C}
        \toprule
        Circuit ID & \text{Parameters} & \text{Depth} & \text{Gates} & \text{Expressibility} & \text{Entanglement} & \mathcal{L}_{PQC} \\
        \midrule
        1 & \mathbf{3} & \mathbf{4} & 8 & 0.601 \pm 0.332  &  \mathbf{1.0 \pm 0.0} & 1.019  \\
        2 & 4 & \mathbf{4} & \mathbf{6} & 0.17 \pm 0.128  &  0.377 \pm 0.0 & 1.439  \\
        3 & 6 & \mathbf{4} & 8 & 0.017 \pm 0.013  &  0.531 \pm 0.006 & 1.203 \\
        4 & 6 & 5 & 9 & 0.046 \pm 0.05  &  0.376 \pm 0.001 & 1.439  \\
        5 & 6 & 8 & 12 & 0.135 \pm 0.142  &  0.627 \pm 0.002 & 1.063  \\
        6 & 12 & 8 & 12 & \mathbf{0.006 \pm 0.001}  &  0.395 \pm 0.002 & 1.407  \\
        AS & \mathbf{3} & 6 & 10 & 0.072 \pm 0.015  &  0.888 \pm 0.0 & \mathbf{1.002}  \\
        \midrule
        \midrule
        \makecell{ Thresholds } & \leq3 & \leq9 & \leq15 & \leq0.02 & \geq0.667 & \leq1.0 \\
        \bottomrule
    \end{tabular}
    }
    \caption{Evaluation of PQCs for the regular QCNN architecture with 3 qubits. The table details each circuit's complexity, expressibility and entanglement (mean $\pm$ standard deviation), and objective function value ($\mathcal{L}_{PQC}$, mean). Bold numbers indicate the best entries in each column.  The bottom row (Thresholds) shows the constraints and objectives of the ansatz search (AS).
    }
    \label{tab:rqcnn_pqc_3q}
\end{table}

\clearpage

\begin{figure}[ht]
    \centering
        \resizebox{\textwidth}{!}{  
    \begin{quantikz}[rounded corners, column sep=5pt, row sep=4pt]
    &  \gate{Z} & & & & & \gate{Y} & \gate{R_z(\theta_2)} & & & & & \\
    &  \gate{R_z(\theta)}  & \gate[2]{ECR} \gateinput{0} & \gate{X} & & \gate{R_y(\theta_0)} & \gate{U(\theta_0,\theta_0,\theta_0)} & \ctrl{-1} & \gate{Y} & \gate{R_y(\theta_2)} & \gate{R_z(\theta_1)} & & \\
    &  & \gateinput{1} & \ctrl{-1} & \gate{U(\theta_0,\theta_0,\theta_0)} & \ctrl{1} & \gate{Z} & \gate{R_z(\theta_1)} & \gate{R_x(\theta_2)} & & \gate{U(\theta_2, \theta_2, \theta_1)} & \gate{\sqrt{X}} & \\
    &  \gate{U(\theta_0,\theta_0,\theta_0)} & & & & \gate{X} & & & & \ctrl{-2} & \gate{R_y(\theta_3)} & & 
    \end{quantikz}
    }
    \caption{
    4-qubit circuit from the ansatz search of the regular QCNN.
    }
    \label{fig:rqcnn_pqc_search_circ_4q}
\end{figure}

    \begin{table}[ht]
        \centering
        \resizebox{\textwidth}{!}{  
        \begin{tabular}{c C C C C C C}
            \toprule
            Circuit ID & \text{Parameters} &  \text{Depth} & \text{Gates} & \text{Expressibility} & \text{Entanglement} & \mathcal{L}_{PQC} \\
            \midrule
            1 & \mathbf{4} & \mathbf{5} & 11 & 0.55 \pm 0.207  &  \mathbf{1.0 \pm 0.0} & 1.008 \\
            2 & 7 & \mathbf{5} & \mathbf{10} & 0.019 \pm 0.014  &  0.367 \pm 0.003 & 1.554 \\
            3 & 8 & \mathbf{5} & 11 & \mathbf{0.007 \pm 0.002}  &  1.621 \pm 0.002 & 1.246 \\
            4 & 8 & 6 & 12 & 0.048 \pm 0.033  &  0.372 \pm 0.002 & 1.549 \\
            5 & 8 & 9 & 16 & 0.01 \pm 0.004  &  0.713 \pm 0.001 & 1.134 \\
            AS & \mathbf{4} & 8 & 19 & 0.039 \pm 0.012  &  0.874 \pm 0.0 & \mathbf{1.0} \\
            \midrule
            \midrule
            \makecell{ Thresholds } & \leq4 & \leq12 & \leq20 & \leq0.019 & \geq0.824 & \leq1.0 \\
            \bottomrule
        \end{tabular}
        }
        \caption{Evaluation of PQC for the regular QCNN architecture with 4 qubits. The table details each circuit's complexity, expressibility and entanglement (mean $\pm$ standard deviation), and objective function value ($\mathcal{L}_{PQC}$, mean). Bold numbers indicate the best entries in each column.  The bottom row (Thresholds) shows the constraints and objectives of the ansatz search (AS).}
        \label{tab:rqcnn_pqc_4q}
    \end{table}

\clearpage

    \begin{figure}[ht]
        \centering
        \begin{tikzpicture}
        \node[rotate=90] {
        \resizebox{1.34\textwidth}{!}{
            \begin{quantikz}[rounded corners, column sep=5pt, row sep=4pt]
             & \gate{X} & & & & \gate{Y} & \gate{R_y(\theta_4)} & & & & & & & &\gate[4]{ECR} \gateinput{0} & \gate{\sqrt{X}} & \ctrl{7} & & \ctrl{5} & & \gate{Y} & & & & 
             \\
            & \gate{U(\theta_1,\theta_2, \theta_3)} & \qw & \gate[8]{ECR} \gateinput{0} & \gate{\sqrt{X}} & & & & & \gate[6]{ECR}\gateinput{0} & & \gate{\sqrt{X}} & & & & & & & & \gate[7]{ECR}\gateinput{0} & & & & \gate{X} &
            \\
            & & & & & & & & & & & & & & & & & \gate[6]{ECR}\gateinput{1} & & & & & & &   
            \\
            & \gate{R_z(\theta_2)} & & & & & & & & & & & \gate{R_x(\theta_3)} & \gate{\sqrt{X}} & \gateinput{1} & & & & & & & &
            \ctrl{3} & \ctrl{-2} &
            \\
            & \gate{R_y(\theta_0)} & \gate{R_x(\theta_1)} & & \ctrl{2} & & \gate[2]{ECR}\gateinput{1} & & \ctrl{3} & & & \gate{R_x(\theta_6)} & & & & & & & & & & & & & 
            \\
            & \gate{H} & & & & & \gateinput{0} & & & & & & & & & & & & \gate{X} & & & \gate{R_z(\theta_6)} & & & 
            \\
            & & & & \gate{Y} & & \gate{H} & \gate{U(\theta_1,\theta_4,\theta_5)} & & \gateinput{1} & & \gate{\sqrt{X}} & \ctrl{-3} & \gate{R_z(\theta_6)} & & & & & & & \ctrl{-6} & & \gate{R_x(\theta_7)} & &
            \\
            & & & & & & & & \gate{R_x(\theta_4)} & & & & & & & & \gate{X} & \gateinput{0} & & \gateinput{1} & & \ctrl{-2} & & & 
            \\
            & & & \gateinput{1} & & \ctrl{-8} & \gate{R_y(\theta_4)} & \gate{R_z(\theta_8)} & & & & & & & & & & & & & & & & &        
            \end{quantikz}
            }};
        \end{tikzpicture}
        
        \caption{9-qubit circuit from the ansatz search of the regular QCNN.}
    \end{figure}

\clearpage

    \begin{table}[ht]
        \centering
        \resizebox{\textwidth}{!}{  
        \begin{tabular}{c C C C C C C}
            \toprule
            Circuit ID & \text{Parameters} &  \text{Depth} & \text{Gates} & \text{Expressibility} & \text{Entanglement} & \mathcal{L}_{PQC} \\
            \midrule
            1 & \mathbf{9} & 10  & 26 & inf  &  \mathbf{1.0 \pm 0.0} & inf  \\
            2 & 16 & \mathbf{8} & \mathbf{24} & 0.005 \pm 0.002  &  0.435 \pm 0.002 & 1.562  \\
            3 & 18 & 10 & 26 & \mathbf{0.001 \pm 0.001}  &  0.792 \pm 0.001 & 1.203  \\
            4 & 18 & 11 & 27 & \mathbf{0.001 \pm 0.001}  &  0.377 \pm 0.001 & 1.621 \\
            5 & 18 & 13 & 36  & \mathbf{0.001 \pm 0.001}  &  0.762 \pm 0.001 & 1.234  \\
            AS & \mathbf{9} & 18 & 45 & 0.047 \pm 0.006  &  \mathbf{1.0 \pm 0.0} & \mathbf{1.0}  \\
            \midrule
            \midrule
            \makecell{ Thresholds } & \leq9 & \leq27 & \leq45 & \leq0.013 & \geq0.994 & \leq1.0  \\
            \bottomrule
        \end{tabular}
        }
        \caption{Evaluation of PQC for the regular QCNN architecture with 9 qubits. The table details each circuit's complexity, expressibility and entanglement (mean $\pm$ standard deviation), and objective function value ($\mathcal{L}_{PQC}$, mean). Bold numbers indicate the best entries in each column.  The bottom row (Thresholds) shows the constraints and objectives of the ansatz search (AS). Circuit 1 has an infinite expressibility value, generating state pairs with fidelities that have a $10^{-30}$ probability of occurring when sampling from the uniform (Haar) measure.}
        \label{tab:rqcnn_pqc_9q}
    \end{table}